\documentclass{ecai} 

\usepackage{latexsym}
\usepackage{amssymb}
\usepackage{amsmath}
\usepackage{amsthm}
\usepackage{booktabs}
\usepackage{enumitem}
\usepackage{graphicx}
\usepackage{color}
\usepackage{pifont}
\usepackage[bibliography=common]{apxproof}
\newtheoremrep{theorem}{Theorem}
\usepackage{multicol}
\usepackage{refcount}

\usepackage{algorithm}
\usepackage{algorithmic}

\usepackage{cleveref}
\crefname{theorem}{Theorem}{Theorems}
\crefname{proposition}{Proposition}{Propositions}
\crefname{lemma}{Lemma}{Lemmas}
\crefname{exmp}{Example}{Examples}
\crefname{corollary}{Corollary}{Corollaries}
\crefname{claim}{Claim}{Claims}
\crefname{remark}{Remark}{Remarks}
\crefname{section}{Section}{Sections}
\crefname{definition}{Definition}{Definitions}
\crefname{example}{Example}{Examples}
\crefname{table}{Table}{Tables}
\crefname{figure}{Figure}{Figures}
\crefname{appendix}{Appendix}{Appendices}
\crefname{equation}{Equation}{Equations}
\crefname{algorithm}{Algorithm}{Algorithms}
\crefname{subsection}{Subsection}{Subsections}

\newcommand{\student}{i}
\newcommand{\studentset}{I}
\newcommand{\school}{s}
\newcommand{\schoolset}{S}
\newcommand{\instance}{M}
\newcommand{\matching}{Y}

\newtheorem{lemma}[theorem]{Lemma}

\newtheorem{definition}{Definition}
\newtheorem{example}{Example}
\newtheorem{remark}{Remark}

\newcommand{\BibTeX}{B\kern-.05em{\sc i\kern-.025em b}\kern-.08em\TeX}

\newif\ifconference
\conferencefalse

\begin{document}


\begin{frontmatter}


\paperid{6040} 


\title{
A New Relaxation of Fairness in Two-Sided Matching Respecting Acquaintance Relationships
}


\author[A]{\fnms{Ryota}~\snm{Takeshima}}
\author[A]{\fnms{Kei}~\snm{Kimura}\orcid{0000-0002-0560-5127}\thanks{Corresponding Author. Email: kkimura@inf.kyushu-u.ac.jp}}
\author[A]{\fnms{Ayumu}~\snm{Kuroki}}
\author[A]{\fnms{Temma}~\snm{Wakasugi}}
\author[A]{\fnms{Makoto}~\snm{Yokoo}\orcid{0000-0003-4929-396X}} 

\address[A]{Kyushu University}


\begin{abstract}
Two-sided matching, such as matching between students and schools, has been applied to various aspects of real life and has been the subject of much research, however, it has been plagued by the fact that efficiency and fairness are incompatible.
In particular, Pareto efficiency and justified-envy-freeness are known to be incompatible even in the simplest one-to-one matching, i.e., the stable marriage problem.
In previous research, the primary approach to improving efficiency in matchings has been to tolerate students' envy, thereby relaxing fairness constraints.
In this study, we take a different approach to relaxing fairness.
Specifically, it focuses on addressing only the envy that students may experience or prioritize more highly and seeks matchings without such envy.
More specifically, this study assumes that envy towards students who are not acquaintances has less impact compared to envy towards students who are acquaintances.
Accordingly, we assume that the students know each other or not, represented by an undirected graph, and define a local envy as a justified envy toward an acquaintance or a neighbor in the graph.
We then propose the property that there is no local envy as a new relaxed concept of fairness, called local envy-freeness.
We analyze whether Pareto-efficient matching can be achieved while maintaining local envy-freeness by meaningfully restricting the graph structure and the school's preferences.
To analyze in detail the fairness that can achieve Pareto-efficient matching, we introduce a local version of the relaxed fairness recently proposed by Cho et al. (AAMAS 2024), which parameterizes the level of local envy-freeness by nonnegative integers.
We then clarify the level of local envy-freeness that can be achieved by Pareto-efficient mechanisms for graphs that are ``close'' to trees and single-peaked preferences on the graphs.
\end{abstract}

\end{frontmatter}


\section{Introduction}
\label{sec:intro}
Two-sided matching, 
such as matching between students and schools, residents and hospitals, and workers and firms, has been applied to various situations in real life, and has been actively studied in various fields, such as artificial intelligence, economics, and operations research.
Extensive research has been conducted to find matchings with desirable properties~\cite{Kesten10,hosseini2015manipulablity,kawase2017near,Yahiro18,Haris19matching,IsmailiHZSY19,DGY19,suzuki2022strategyproof,CHK24}, where efficiency and fairness 
are regarded as particularly important and desirable properties.

This study focuses particularly on two-sided matching in the school choice problem, i.e., matchings between students and schools.
The Deferred Acceptance (DA) mechanism~\cite{Gale:AMM:1962} is well-known for satisfying stability, i.e., weak efficiency (that is, nonwastefulness) and fairness (that is, justified-envy-freeness), and its variants are frequently used in practical applications. 
Furthermore, the DA mechanism is known to always produce student-optimal stable matchings among all the stable matchings. However, it has been revealed that these student-optimal matchings, while strictly meeting fairness criteria, can significantly compromise efficiency~\cite{Kesten10}.
Indeed, it is widely recognized that there is a trade-off between efficiency and fairness.
Let us explain this in detail.
Pareto efficiency is considered the strongest concept of efficiency.
Intuitively, a matching is \emph{Pareto efficient (PE)} if no subset of students can be matched with better schools without hurting the other students.
Justified-envy-freeness is the most widely studied concept of fairness.
Intuitively, a student has a \emph{justified envy} towards another student if the former prefers to be matched with the school matched to the latter and the school also prefers to be matched with the former to the latter (in other words, the school has a higher priority for the former to the latter).
A matching is \emph{justified-envy-free} (or often called \emph{fair}) if no student has a justified envy.
It is known that Pareto efficiency and justified-envy-freeness is incompatible even in the one-to-one two-sided matching, i.e., the stable marriage problem~\cite{roth1982economics,AS03}.

For the above reasons, considerable research has been conducted to seek highly efficient matchings by partially sacrificing fairness.
A particularly notable study by Kesten~\cite{Kesten10} proposes a method where students agree to relinquish their priority (or envy), as long as it does not affect their matched school, to enhance the efficiency of other students.
\ifconference
\else
This method is known as an Efficiency-Adjusted Deferred Acceptance Mechanism (EADAM).
\fi
Dur et al.~\cite{DGY19} conduct research in the setting where each school determines which priority violations are permissible.
These methods require asking students and schools about abandonment of envy and priority.

This paper proposes a different approach to relaxing fairness compared to the aforementioned studies.
Specifically, it focuses on addressing only the envy that has more impact on students, and seeks matchings without such envy.
More specifically, this study assumes that envy towards students who are not acquaintances has less impact compared to envy towards students who are acquaintances.
Therefore, we consider a matching where there is no envy towards acquaintances to be a desirable matching.

More precisely, 
we assume that the students know each other or not, 
and define a \emph{local envy} as a justified envy toward an acquaintance.
We then propose the property that there is no local envy as a new relaxed concept of fairness, called \emph{local envy-freeness}.
If all students know each other, then local envy-freeness coincides with the standard fairness (i.e., justified-envy-freeness).
So, our model is a generalization of the ordinary one.
In this paper, we theoretically analyze how much efficiency can be improved by introducing local envy-freeness.

First, we unfortunately demonstrate that fundamental theorems of stable matching do not hold in this extended model.
Specifically, we prove that results such as the set of stable matchings forming a lattice~\cite{Roth:CUP:1990}, the existence of a student-optimal stable matching within the set of stable matchings~\cite{Gale:AMM:1962}, and the rural hospitals theorem~\cite{Roth:1986} do not hold in locally envy-free matching in general.

However, as a positive result, we demonstrate that by imposing natural preference restrictions and acquaintance constraints in school choice, it is possible to achieve Pareto efficient and locally envy-free matchings.
More specifically, we represent acquaintance relationships among students as an undirected graph, referred to as a \emph{student acquaintance graph}, and analyze what kind of matchings can be obtained based on the structure of this graph.

We start by showing that even if the student acquaintance graph is a path (in other words, when students are arranged in a straight line and only neighboring students know each other), there may be no Pareto efficient and local envy-free matching (\cref{thm:no-LEF-PE-under-path}).
Therefore, we consider restricting the schools' preferences in a meaningful way.
Specifically, we show that there always exists a Pareto efficient and locally envy-free matching for any path graph when the schools' preferences are single-peaked with respect to the path (\cref{thm:LEF-PE-under-tree-single-peaked}).
Here, a school's preference is \emph{single-peaked} if there is a unique peak (most preferred student) and the preference decreases as one moves away from this peak along the path.
This models a situation where 
and 
a school prefers students who live closer to the school.

In fact, we demonstrate that the above results hold not only for path graphs but also for tree graphs (that is, connected graphs that have no cycles). 
This allows us to handle not only the case where students' homes are arranged in a straight line but also the more general and realistic case where they are arranged on a tree representing a road network.
Here, single-peaked preferences on a tree are defined in the same way as single-peaked preferences on a path.

\begin{table}[t]
\caption{Existence of Pareto efficient, locally ERF-$\ell$, and strategyproof mechanism for graphs with treewidth $k$ 
(\ding{52} means such a mechanism exists for all $k$, 
\ding{51} means such a mechanism does not exist for some $k$ but exists for all $k$ without strategyproofness, and 
\ding{54} means even without strategyproofness, 
a matching that satisfies Pareto efficiency and the fairness property 
may not exist for some $k$.) }
\label{tbl:results}
\begin{center}
\begin{tabular}{c|lll}
schools' &\multicolumn{3}{c}{level $\ell$ of local ERF}\\
\cline{2-4}
preferences & $k-2$ ($k \ge 2$) & $k-1$ & $k$ \\
\hline 
\hline 
general & \ding{54} [Thm~\ref{thm:no-LEF-PE-under-cycle-single-peaked}] & \ding{54} [Thm~\ref{thm:no-LEF-PE-under-path}] & \ding{52} [Thm~\ref{thm:LreverseEFk-PE-under-treewidth-k}]\\
single-peaked & \ding{54} [Thm~\ref{thm:no-LEF-PE-under-cycle-single-peaked}] & \ding{51} [Thm~\ref{thm:no-SP-LEF-PE-under-path-single-peaked},\ref{thm:LEF-k-1-PE-under-k-degenerate-single-peaked}] & \ding{52} [Thm~\ref{thm:LreverseEFk-PE-under-treewidth-k}] \\
\hline
\end{tabular}
\end{center}
\end{table}

To broaden the applications of our model, we also analyze what happens when the structure of the graph is made more general.
Unfortunately, we reveal that Pareto efficiency and local envy-freeness are incompatible in the simplest non-tree graph, that is, the cycle of length three (\cref{thm:no-LEF-PE-under-cycle-single-peaked}). 
Therefore, by further relaxing fairness constraints, we enable the pursuit of efficient matchings.

We propose new relaxed notions of fairness, based on such a notion recently proposed by~\citet{cho:2024}.
The weaker fairness requirement introduced in~\cite{cho:2024} is called \emph{envy-freeness up to $k$ peers} (EF-$k$), where EF-$k$ guarantees that each student has justified envy toward at most $k$ students.
By varying $k$, EF-$k$ can represent different levels of fairness.
We define a local version of EF-$k$ as \emph{local EF-$k$}.
Local EF-$k$ guarantees that each student has justified envy toward at most $k$ \emph{neighbor} students.
We also define a variant of EF-$k$, which we call \emph{envy-receiving-freeness up to $k$ peers} (ERF-$k$), where ERF-$k$ guarantees that each student receives justified envy \emph{from} at most $k$ students, as well as its local version \emph{local ERF-$k$}.

We then extend our results for tree graphs to graphs that are ``tree-like'' in two ways.

First, we consider a graph with an underlying tree to which edges are added such that the degree (that is, the number of neighboring vertices) of each vertex is at most $k$ with an assumption that the schools' preferences are single-peaked on the underlying tree.
For this setting, there might be no Pareto efficient and local envy-free matching, but we show there exists a mechanism that finds a matching that satisfies Pareto efficiency, local EF-($k-1$), and local ERF-($k-1$).

Second, we consider graphs with treewidth $k$, where $k$ is a positive integer.
Treewidth is a graph-theoretic parameter that quantifies how similar a given graph is to a tree (see \cref{def:treewidth} in \cref{subsec:treewidthk-LEF-k-1} for a formal definition).
In fact, a graph is a tree if and only if it has treewidth one and connected.
We extend the single-peaked preferences on trees to those on treewidth $k$ graphs and propose a mechanism that satisfies Pareto efficiency and local ERF-($k-1$) when the graph has treewidth $k$ and the schools' preferences are single-peaked on the graph.
Finally, for graphs with treewidth $k$ and general schools' preferences, we propose a mechanism based on the serial dictatorship (SD)~\cite{AS98} that satisfies strategyproofness, Pareto efficiency, and local ERF-$k$. 
From the impossibility theorems above, our mechanisms achieve the best possible fairness guarantees in terms of the number of students by whom each student is locally envied.
Our results are summarized in \cref{tbl:results}.

By developing a Pareto efficient and less envious mechanism not only for tree graphs but also for graphs that are close to trees (e.g., graphs with low treewidth), we anticipate that the following applications can be significantly broadened.
First, even if the road network of the town where the students reside is not a tree, our mechanism can still achieve Pareto-efficient matching with only a slight compromise in fairness, provided that the network is sufficiently close to a tree structure.
Furthermore, we posit that when schools determine their preferences for students based on criteria other than geographical proximity, such preferences can often be represented as single-peaked preferences on graphs that are close to trees (see Section~\ref{sec:conclusion} for further details).

\subsection{Related work}
There are several other studies that focus on local stability and local fairness~\cite{Arcaute2009,AIK+13,ABH13,CM13,HW17,KSM19}.
The research particularly relevant to our study is that of \citet{Arcaute2009}, who introduced local stability in the context of matching in job markets (that is, matchings between workers and firms).
This concept of stability is based on the assumption that a worker can only make an offer to a firm through an acquaintance who is already employed there. 
Therefore, a local concept is proposed from a different motivation than ours. 
Mathematically, their local stability and our local envy-freeness coincide in the case of one-to-one matching, but generally differ in the case of many-to-one matching. 
These comparisons are detailed in \cref{sec:Comparison-LS-and-LEF}.

Famous Pareto efficient mechanisms such as the serial dictatorship (SD)~\cite{AS98} and the top trading cycle (TTC)~\cite{AS03} can result in very poor fairness. 
Therefore, we believe that the results of this study, which provide theoretical guarantees for fairness, can be useful in certain situations.

The setting in which agents can hold envy only for their neighbors in a graph has been studied so far in fair division~\cite{Abebe2017,Bredereck}, such as house allocation~\cite{Beynier2019}, cake cutting~\cite{Bei2017}, and combinatorial auction~\cite{FLAMMINI20191}.
To the best of our knowledge, this paper is the first to address this type of local concept regarding fairness in many-to-one two-sided matching.

Finally, single-peaked preferences (on a path or a tree) frequently appear in voting theory and facility location. Voters' preferences for candidates are often single-peaked, and assuming single-peaked preferences can lead to various theoretically desirable properties~\cite{HANSEN19811,DEMANGE1982389,TRICK1989329,ELP17}.
In facility location problems, it is common to assume single-peaked preferences for the placement of public goods such as libraries~\cite{Moulin80}.

\section{Model}
For a positive integer $z$ we denote $\{1, \dots, z\}$ by $[z]$.

\subsection{Matching market with acquaintance graphs}
\label{sec:model}
A matching market with a student acquaintance graph 
is given by $\instance=(\studentset, \schoolset, X, \succ_\studentset, \succ_\schoolset, q, G)$.\footnote{In models of school choice, it is sometimes assumed that schools can accept all students. However, this paper adopts a model in which some students may be unacceptable to certain schools---a setting that is also applicable to scenarios such as college admissions.}
The meaning of each element is as follows. 
\begin{itemize}
 \item $\studentset=\{\student_1, \ldots, \student_n\}$ is a finite set of students.
 \item $\schoolset=\{\school_1, \ldots, \school_m\}$ is a finite set of schools. 
 \item $X \subseteq \studentset\times \schoolset$ is a finite set of contracts.
Contract $x = (\student, \school) \in X$ 
represents the matching between student $\student$  and school $\school$.
\item For any $Y \subseteq X$, 
let $Y_\student:=\{(\student, \school) \in Y \mid \school \in
\schoolset\}$ and $Y_\school:=\{(\student, \school) \in Y \mid \student \in \studentset\}$ denote the sets of contracts in $Y$ that involve $\student$ and $\school$, respectively.
\item 
${\succ_\studentset} = (\succ_{\student_1}, \ldots, \succ_{\student_n})$ is a profile of
the students' preferences. 
Each student $\student$ has a strict preference order $\succ_{\student}$ over $X_\student \cup {(\student, \emptyset)}$, where $X_\student=\{(\student, \school) \in X \mid \school \in
\schoolset\}$ and $(\student, \emptyset)$ denotes the assignment in which $\student$ is unmatched.
We assume $\succ_\student$ is strict for each $\student$.
A contract $(\student, \school)$ is \emph{acceptable} for $\student$ if $(\student, \school) \succ_\student (\student, \emptyset)$.
Occasionally, we use notations like $\school \succ_\student \school'$ instead of $(\student,\school) \succ_\student (\student,\school')$ for brevity.
\item ${\succ_\schoolset} = (\succ_{\school_1}, \ldots, \succ_{\school_m})$ is a profile of schools' preferences. 
Each school $\school$ has a strict preference order $\succ_\school$ over $X_\school \cup {(\emptyset, \school)}$, 
where $(\emptyset, \school)$ represents the option in which $\school$ leaves one of its seats vacant.
A contract $(\student, \school)$ is \emph{acceptable} for $\school$ if $(\student, \school) \succ_\school (\emptyset, \school)$.
We sometimes use notations like $\student \succ_\school \student'$ instead of $(\student, \school) \succ_\school (\student', \school)$ for brevity.
\item $q \in {\mathbb{Z}}_+^m$ is a vector that represents schools' maximum quotas, where $m$ is the number of schools and ${\mathbb{Z}}_+^m$ is the set of vectors of $m$ non-negative integers.
\item $G = (\studentset,E)$ is an undirected graph, called a \emph{student acquaintance graph}, which models acquaintance relationships among students.
Here, $\studentset$ is a set of vertices (students) and $E$ is a set of edges, where edge $\{ \student,\student' \} \in E$ means that $\student$ and $\student'$ know each other.
\end{itemize}
For graph $G = (\studentset,E)$ we denote by $N(\student)$ the set of neighbors of vertex $\student \in \studentset$, i.e., $N(\student) = \{ \student' \in \studentset \mid \{ \student,\student' \} \in E \}$.

We say $Y\subseteq X$ is 
a \emph{matching}, 
if for each $\student \in \studentset$, 
$|Y_\student| \le 1$ holds.

\begin{definition}[feasibility]
\label{def:standard}
We say $Y \subseteq X$ is \emph{feasible} if $\forall j \in [m]$, $|Y_{\school_j}| \leq q_{\school_j}$ holds.
\end{definition}

With a slight abuse of notation, for two sets of contracts
$Y$ and $Y'$,
we denote $Y_\student \succ_\student Y'_\student$ if either (i)
$Y_\student = \{x\}$, $Y'_\student = \{x'\}$, and $x \succ_\student x'$
for some $x, x' \in X_\student$,
or
(ii) $Y_\student = \{x\}$ for some $x \in X_\student$, $Y'_\student = \emptyset$, and $x \succ_\student (\student,\emptyset)$.
Furthermore, we denote $Y_\student \succeq_\student Y'_\student$ if  either $Y_\student \succ_\student Y'_\student$
or $Y_\student = Y'_\student$.
Also, we use notations like  $x \succ_\student Y_\student$ or 
$\school \succ_\student Y_\student$, where $x$ is a contract, $Y$ is a matching, and $\school$ is a school.

We will prove several impossibility results on matching markets with student acquaintance graphs for general preference profiles.
Therefore, we also consider restricting the schools' preferences.

\begin{definition}[single-peaked preferences on trees~\cite{DEMANGE1982389,Endriss:2017}]
Assume that graph $G$ is a tree, that is, it is connected and has no cycle.
A preference $\succ_\school$ over the set $\studentset$ of students is called \emph{single-peaked on $G$} if for any $k \in [n]$, the set of top $k$ students (except $\emptyset$) under $\succ_\school$ is connected in $G$.
\end{definition}

For example, a single-peaked preference on a tree can model a situation where 
students' homes are arranged in a tree road network, neighboring students know each other, and 
a school prefers students who live closer to the school.

\subsection{Desirable properties for matchings}

Let us introduce several desirable properties of a matching and a
mechanism.
We say a mechanism satisfies property A if the mechanism produces a matching that satisfies property A in every possible matching market. 

First, we define fairness. 
\begin{definition}[fairness]
\label{def:fairness}
In matching $Y$, student $\student$ \emph{has justified envy toward
another student $\student'$} if 
$\school \succ_\student Y_\student$, $(\student', \school) \in Y$, and
$\student \succ_\school \student'$ hold.
We say matching $Y$ is \emph{fair} or \emph{justified-envy-free} if no student has justified envy.
\end{definition}
Fairness implies that if student $\student$ is not assigned to school $\school$ (although she hopes to be assigned), 
then $\school$ prefers all students assigned to it over $\student$. 

Now, we introduce the concept of local envy-freeness.

\begin{definition}[local envy-freeness]
\label{def:local-fairness}
Let $G = (\studentset,E)$ be a student acquaintance graph.
In matching $Y$, student $\student$ \emph{has local (justified) envy toward
another student $\student'$ with respect to $G$} if $\student$ has justified envy toward $\student'$ and in addition $\student$ and $\student'$ are adjacent in $G$ (i.e., $\{\student,\student'\} \in E$).
We say matching $Y$ is \emph{locally envy-free with respect to $G$} if no student has local envy with respect to $G$.
\end{definition}

Note that by definition if student acquaintance graph $G$ is a complete graph, 
then local envy-freeness with respect to $G$ coincides with fairness.
Local envy-freeness implies that students do not feel envy towards their acquaintances.

We often say matching $Y$ is \emph{locally envy-free (LEF)}, omitting ``with respect to $G$,''  if $G$ is clear from the context.

We also define a weaker fairness concept called envy-free up to $k$ peers (EF-$k$) introduced by~\citet{cho:2024}.
For matching $Y$ and student $\student$, let $Ev(Y, \student)$ denote 
$\{\student' \mid \student'\in \studentset, \student \text{ has } \text{justified envy toward } \student' \text{ in } Y\}$.

\begin{definition}[envy-free up to $k$ peers]
Matching $Y$ is \emph{envy-free up to $k$ peers (EF-$k$)} if $\forall \student \in \studentset$, 
$|Ev(Y,\student)|\leq k$ holds. 
\end{definition}
EF-$0$ is equivalent to fairness. 
Any matching is EF-$(n-1)$, where $n=|\studentset|$. 

Here, we introduce a new relaxation of fairness that complements EF-$k$.
A student who is being envied by a large of other students is considered to be significantly unfairly allocated to the school. 
Therefore, placing a limit on the number of envy received is also deemed desirable from the perspective of fairness.
Motivated by this consideration, we introduce a weaker fairness concept called \emph{envy-receiving-free up to $k$ peers (ERF-$k$)}.
For matching $Y$ and student $\student$, let $Evr(Y, \student)$ denote 
$\{\student' \mid \student'\in \studentset, \student' \text{ has } \text{justified envy toward } \student \text{ in } Y\}$.

\begin{definition}[envy-receiving-free up to $k$ peers]
Matching $Y$ is \emph{envy-receiving-free up to $k$ peers (ERF-$k$)} if $\forall \student \in \studentset$, 
$|Evr(Y,\student)|\leq k$ holds. 
\end{definition}

\ifconference
\else
ERF-$0$ is equivalent to fairness. 
Any matching is ERF-$(n-1)$, where $n=|\studentset|$. 
\fi 

What kind of relaxed fairness concept is desirable, in particular, whether EF-$k$ or ERF-$k$ is more appropriate, would certainly be debatable and would depend on the specific application domain. 
That being said, 
 we believe that satisfying ERF-$k$ is more important in the context of school choice. 
 In particular, when a student is envied by many others, it typically indicates that the student has very low priority at the school to which she was assigned. Assigning such a student to that school may be undesirable, as it can strongly undermine the perceived fairness of the system. 
 In contrast, a student who envies many others may still have been assigned to a school that is not particularly bad from their own perspective. Therefore, we believe that avoiding the former situation is more critical than the latter.

Now, we introduce local concepts of these weaker fairness notions.
For matching $Y$ and student $\student$, let $loc{\text -}Ev(Y, \student)$ denote 
$\{\student' \mid \student'\in N(\student), \student \text{ has } \text{justified envy toward } \student' \text{ in } Y\}$.

\begin{definition}[locally envy-free up to $k$ peers]
Matching $Y$ is \emph{locally envy-free up to $k$ peers (locally EF-$k$)} if $\forall \student \in \studentset$, $|loc{\text -}Ev(Y,\student)|\leq k$ holds. 
\end{definition}

Analogously, for matching $Y$ and student $\student$, let $loc{\text -}Evr(Y, \student)$ denote 
$\{\student' \mid \student'\in N(\student), \student' \text{ has } \text{justified envy toward } \student \text{ in } Y\}$.

\begin{definition}[locally envy-receiving-free up to $k$ peers]
Matching $Y$ is \emph{locally envy-receiving-free up to $k$ peers (locally ERF-$k$)} if $\forall \student \in \studentset$, 
$|loc{\text -} Evr(Y,\student)|\leq k$ holds. 
\end{definition}

We also introduce a relaxation of fairness proposed by \citet{Toda2006}, which is satisfied by some of our mechanisms.

\begin{definition}[mutually-best]
\label{def:mutually-best}
We say a pair of student $\student$ and school $\school$ is a \emph{mutually-best pair (MB-pair)} if $\school \succ_\student \school'$ for any $\school' \in \schoolset \setminus \{\school\}$ and $\student \succ_\school \student'$ for any $\student' \in \studentset \setminus \{\student\}$.
Matching $Y$ is \emph{mutually-best (MB)}, if any MB-pair $(\student,\school)$ is matched.
\end{definition}
By definition, fairness implies MB.

Next, we define two of the most extensively studied properties on student welfare (efficiency).

\begin{definition}[Pareto efficiency]
    Matching $Y$ is Pareto dominated by another 
    matching $Y'$ if $\forall \student \in \studentset$, $Y'_\student \succeq_\student Y_\student$, and
    $\exists \student \in \studentset$,  $Y'_\student \succ_\student Y_\student$ hold.
	Feasible matching $Y$ is Pareto efficient if
	no other feasible matching Pareto dominates it.
\end{definition}

\begin{definition}[nonwastefulness]
\label{def:nonwastefulness}
In matching $Y$, student $\student$ \emph{claims an empty seat} of school $\school$
if 
$(\student,\school)$ is acceptable for $\student$,
$c \succ_\student Y_\student$,
and $(Y \setminus Y_\student) \cup \{(\student, \school)\}$ is feasible.
We say feasible matching $Y$ is \emph{nonwasteful} if no student claims an empty
seat. 
\end{definition}
Intuitively, nonwastefulness means that we cannot improve the matching of one student without affecting other students. 
Hence, nonwastefulness is implied by Pareto efficiency.

A feasible matching is \emph{stable} if it is both fair (i.e., justified-envy-free) and nonwasteful.

By definition, stability implies fairness, and fairness implies local envy-freeness. 
In the following, we present an example that illustrates that stability, fairness, and local envy-freeness are distinct concepts.

\begin{example}
This example illustrates that stability, fairness, and local envy-freeness are distinct concepts.
Consider a matching market with two students, $I = \{i_1, i_2\}$, and one school, $S = \{s_1\}$. 
$X$ is defined as $X=I \times S$.
The students’ preference profile $\succ_S$ is given by:  
\[
\begin{array}{c}
    i_1: s_1 \succ_{i_1} \emptyset, \\
    i_2: s_1 \succ_{i_2} \emptyset.
\end{array}
\]
The school's preference profile $\succ_S$ is given by:  
\[
\begin{array}{c}
    s_1: i_1 \succ_{s_1} i_2 \succ_{s_1} \emptyset.
\end{array}
\]
The maximum quota $q_{s_1}$ is set to one, and the student acquaintance graph $G=(I,E)$ is defined with $E=\emptyset$; that is, students $s_1$ and $s_2$ are not acquainted with each other.

Now, one can verify that the set of stable matchings is $\{ \{ (i_1,s_1)\} \}$, 
the set of fair matchings is $\{ \{ (i_1,s_1)\}, \emptyset \}$, and 
the set of locally envy-free matchings is $\{ \{ (i_1,s_1)\}, \{ (i_2,s_1)\}, \emptyset \}$.
\end{example}

Next, we introduce strategyproofness. 
\begin{definition}[strategyproofness]
\label{def:strategyproofness}
We say a mechanism is \emph{strategyproof}
if no student ever has any incentive
to misreport her preference no matter what the other students report. 
More specifically, 
let $Y$ denote the matching obtained when $\student$ declares her true preference $\succ_\student$, 
and $Y'$ denote the matching obtained when $\student$ declare something else, 
then $Y_\student \succeq_\student Y'_\student$ holds. 
\end{definition}
Strategyproofness means that no student has an incentive to misreport her preference. 

Here, we consider strategic manipulations only by students. 
It is well-known that even in the most basic model of one-to-one
matching \cite{Gale:AMM:1962}, satisfying strategyproofness (as well as basic 
fairness and efficiency requirements) for both sides is impossible~\cite{roth1982economics}. 
One rationale for ignoring the school side would be that 
the preference of a school must be
presented in an objective way and cannot be skewed arbitrarily.

Finally, we introduce an existing mechanism that achieves Pareto efficiency.
The Serial Dictatorship (SD) mechanism \cite{AS98,goto:17} is parameterized by an exogenous serial order over the students called a \emph{master-list}. 
We denote the fact that $\student$ is placed in a higher/earlier position than  student $\student'$ in master-list $L$ as $\student \succ_{L} \student'$.
Students are assigned sequentially according to the master-list to their most preferred schools with remaining seats.
SD is strategyproof and achieves Pareto efficiency.

\section{Non-applicability of fundamental theorems in stable matching within locally envy-free matching}
\label{sec:failure-fundamental-theorems}

In this section, we demonstrate that the following fundamental theorems, known to hold in stable matching, do not apply in our extended model.

\begin{itemize}
\item The set of stable matchings forms a lattice~\cite{Roth:CUP:1990}.
\item There exists a student-optimal matching in the set of stable matchings~\cite{Gale:AMM:1962}.
\item The rural hospitals theorem~\cite{Roth:1986}, that is, the set of assigned students, and the number of filled seats in each school, are the same in all stable matchings.
\end{itemize}

Although these mathematically elegant theorems do not hold in our model, 
we demonstrate in \cref{sec:new-mechanisms} that our model yields more efficient matchings in realistic preferences and acquaintance relationships compared to traditional stable matchings.

We first show that the set of locally envy-free (LEF) matchings might not form a lattice.
Recall that a partially ordered set (poset) $(L,\le)$ is called a \emph{lattice} if each two-element subset $\{a,b\}\subseteq L$ has both a join (i.e., least upper bound) 
and a meet (i.e., greatest lower bound).
It is known~\cite{Roth:CUP:1990} that the set of stable matching is a lattice with a partial order $\le$ defined as $\matching_1 \le \matching_2$ if $\matching_1$ is Pareto dominated by $\matching_2$ or $\matching_1 = \matching_2$.

\begin{theorem}
\label{thm:no-lattice-structure}
If the student acquaintance graph is not a complete graph, 
then there exists a matching market in which the set of LEF matchings does not form a lattice with the partial order defined above.
\end{theorem}

\ifconference
\else
\ifconference
\setcounter{theorem}{\getrefnumber{thm:no-lattice-structure}}
\addtocounter{theorem}{-1}
\begin{theorem}
If the student acquaintance graph is not a complete graph, 
then there exists a matching market in which the set of LEF matchings does not form a lattice with the partial order defined above.
\end{theorem}
\begin{proof}
\else
\begin{proof}
\fi
We show that for a specific graph structure and preference profile,
there exist two locally envy-free matchings $Y_1$ and $Y_2$ such that 
\begin{align*}
(*) &\text{
there exists no LEF matching that simultaneously Pareto} \\
&\text{dominates both $Y_1$ and $Y_2$.}
\end{align*}
Let $\studentset = \{\student_1, \student_2, \ldots,  \student_n\}$, $\schoolset = \{\school_1, \school_2, \ldots, \school_m\}$, $X= \studentset \times \schoolset$, and $q_{\school_j} = 1$ for all $j \in [m]$.
The student acquaintance graph $G$ is defined as a complete graph on $\studentset$, 
with a single edge $\{\student_{n-1}, \student_n\}$ removed.

Preferences $\succ_\studentset$ and $\succ_\schoolset$ are given as follows:
\begin{align*}
&\student_j: \school_1 \succ_{\student_j} \school_2 \succ_{\student_j} \ldots \succ_{\student_j} \school_n \qquad(j = 1, 2,\ldots, n)\\
&\school_j: \student_1 \succ_{\school_j} \student_2 \succ_{\school_j} \ldots \succ_{\school_j} \student_n \qquad(j = 1, 2,\ldots, n)
\end{align*}
An LEF matching with respect to $G$ remains LEF even when any single edge is removed from $G$. 
In other words,
if the condition (*) holds with the graph $G$, then the same result can be extended to any graph that is not a complete graph.

We enumerate the two LEF matchings as follows:
\begin{itemize}
\item $Y_1 := \{(\student_1, \school_1),\dots,(\student_{n-2}, \school_{n-2}),(\student_{n-1}, \school_{n-1}),(\student_n, \school_n)\}$
\item $Y_2 := \{(\student_1, \school_1),\dots,(\student_{n-2}, \school_{n-2}),(\student_{n-1}, \school_n),(\student_n, \school_{n-1})\}$
\end{itemize}
These matchings are indeed LEF, since $\student_{n-1}$ and $\student_{n}$ are not adjacent in $G$.

Now, we show that there does not exist a LEF matching that Pareto dominates both $Y_1$ and $Y_2$ simultaneously.
It is well known that the matching obtained by SD mechanism is Pareto efficient~\cite{satterthwaite1981}.
$Y_1$ is the output of the SD mechanism based on the master-list $L_1 = \{\student_1, \student_2, \ldots, \student_{n-1}, \student_n\}$. 
Similarly, $Y_2$ is the output of the SD mechanism based on the master-list $L_2 = \{\student_1, \student_2, \ldots, \student_n, \student_{n-1}\}$.
Therefore, both $Y_1$ and $Y_2$ are Pareto efficient matchings, and there does not exist a matching that Pareto dominates both $Y_1$ and $Y_2$.
\end{proof}
\fi

Note that the fact that the set of stable matchings forms a lattice implies that there exists a unique stable matching that Pareto dominates all other stable matchings in the set. 
This stable matching is referred to as the \emph{student-optimal stable matching}.
We next show that there may not exist a (unique) student-optimal matching in the set of LEF matchings.

\begin{theorem}
\label{thm:no-student-optimal}
There exists a matching market with a student acquaintance graph in which a unique student-optimal matching does not exist in the set of LEF matchings.
\end{theorem}

\ifconference
\else
\ifconference
\setcounter{theorem}{\getrefnumber{thm:no-student-optimal}}
\addtocounter{theorem}{-1}
\begin{theorem}
There exists a matching market with a student acquaintance graph in which a unique student-optimal matching does not exist in the set of LEF matchings.
\end{theorem}
\begin{proof}
\else
\begin{proof}
\fi
By the proof of Theorem~\ref{thm:no-lattice-structure}, there may exist two matchings that are both LEF and PE. 
In such a case, a unique student-optimal matching does not exist.
\end{proof}

\fi

Finally, we show that the rural hospitals theorem~\cite{Roth:1986}, which states that the number of doctors (corresponding to students) assigned to rural hospitals (corresponding to schools) remains unchanged across all stable matchings, does not hold for locally envy-free matching.

\begin{theorem}
\label{thm:no-rural-hospitals}
There exists a matching market with a student acquaintance graph in which the set of assigned students, and the number of filled seats in each school, are different in the LEF matchings.
\end{theorem}

\ifconference
\else
\ifconference
\setcounter{theorem}{\getrefnumber{thm:no-rural-hospitals}}
\addtocounter{theorem}{-1}
\begin{theorem}
There exists a matching market with a student acquaintance graph in which the set of assigned students, and the number of filled seats in each school, are different in the LEF matchings.
\end{theorem}
\begin{proof}
\else
\begin{proof}
\fi
Let $\studentset = \{\student_1, \student_2, \student_3\}$, $\schoolset = \{\school_1, \school_2, \school_3\}$, $X = \studentset \times \schoolset$, and $q_{\school_j} = 1$ for all $j \in [3]$.
Assume that the student acquaintance graph is a path graph depicted in~\cref{fig:path}.

Assume that $\succ_\studentset$ and $\succ_\schoolset$ are given as follows. (Here, preferences for unacceptable schools or students are omitted.)
\begin{align*}
&\student_1: \school_1 \succ_{\student_1} \school_2 \succ_{\student_1} \school_3 & \school_1: \student_2 \succ_{\school_1} \student_1 \succ_{\school_1} \student_3\\
&\student_2: \school_2 & \school_2: \student_3 \succ_{\school_2} \student_2 \succ_{\school_2} \student_1\\
&\student_3: \school_1 \succ_{\student_3} \school_2 \succ_{\student_3} \school_3 & \school_3: \student_1 \succ_{\school_3} \student_2 \succ_{\school_3} \student_3
\end{align*}
We enumerate all the Pareto efficient matchings as follows:
\begin{itemize}
\item $Y_1 := \{$$(\student_1, \school_1)$, $(\student_2, \emptyset)$, $(\student_3, \school_2)$$\}$
\item $Y_2 := \{$$(\student_1, \school_3)$, $(\student_2, \school_2)$, $(\student_3, \school_1)$$\}$
\end{itemize}
the size of $Y_1$ is two, and the size of $Y_2$ is three.
\end{proof}

\fi

\ifconference
\else
\begin{figure}[t]
\centering
\includegraphics[scale = 0.60]{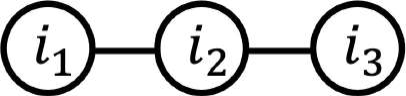}
\hspace{5cm}
\caption{A path graph.}
\label{fig:path}
\hspace{10cm}
\end{figure}
\fi 

From the proof of the above theorem, 
the sizes of LEF and Pareto efficient (PE) matchings may differ.
Now, we show that finding an LEF and PE matching with maximum size is NP-hard.
In fact, we show a stronger result. 
We consider the problem of deciding whether for a given matching market with a student acquaintance graph there exists an LEF and PE matching.
We call this problem the \emph{Locally Envy-free Efficient (LEE)} problem.
We show that the LEE problem is NP-complete.
\ifconference
\else
The proof is similar to the one of Theorem 1 in \cite{AG20}, but it is included for the sake of completeness of the paper.
\fi

\begin{theorem}
\label{thm:existence-LEF+PE-NPC}
The LEE problem is NP-complete.
\end{theorem}

\ifconference
\else
\begin{figure}[t]
\centering
\includegraphics[scale = 0.40]{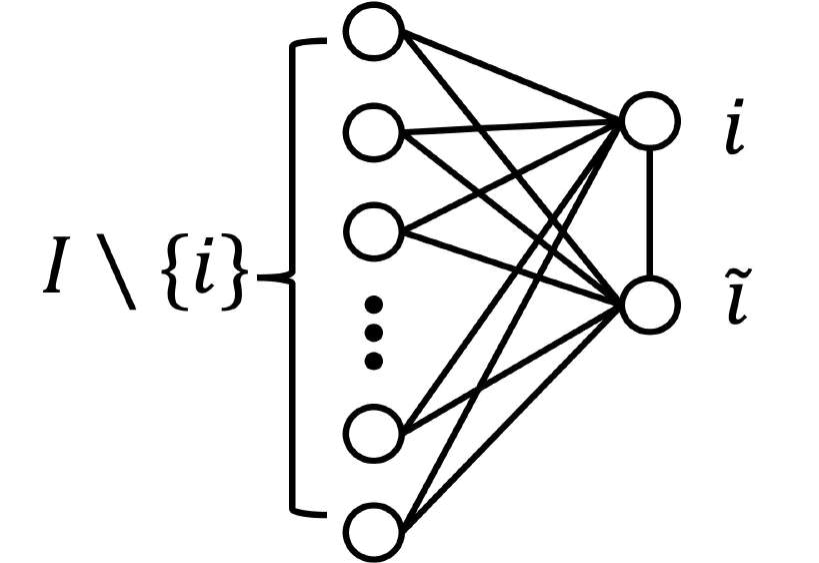}
\hspace{5cm}
\caption{The student acquaintance graph in the proof of~\cref{thm:existence-LEF+PE-NPC}.}
\label{fig:NP-complete-proof}
\hspace{10cm}
\end{figure}
\ifconference
\setcounter{theorem}{\getrefnumber{thm:existence-LEF+PE-NPC}}
\addtocounter{theorem}{-1}
\begin{theorem}
The LEE problem is NP-complete.
\end{theorem}
\begin{proof}
\else
\begin{proof}
\fi
Given an LEF and PE matching, we can check that it is indeed an LEF and PE matching in polynomial time.
Indeed, checking if it is LEF is immediate by definition.
Also, there exists a polynomial time algorithm to decide whether a given matching is PE (\cite[Proposition 3]{ACMM04}).
Hence, the LEE problem is in NP.

Now, we show that the LEE problem is NP-hard.
We show this by reducing the SD feasibility problem, known to be NP-hard~\cite[Theorem 2]{SS15}, to the LEE problem.

Let us first define the SD feasibility problem.
In the SD feasibility problem, we are given a set $\studentset$ of agents, a set $\schoolset$ of objects with $|\schoolset|=|\studentset|$, a preference profile $\succ_{\studentset}$ of the agents over $\schoolset$, and an agent-object pair $(\student,\school) \in \studentset \times \schoolset$.
The task is to decide whether there exists a master-list $L$ such that, when SD is executed based on $L$, agent $\student$ is matched to object $\school$.
This problem is equivalent to determining whether there exists a PE matching in which agent $\student$ is matched to object $\school$.
This equivalence holds because, under maximum quotas constraints, it is known that the set of matchings obtainable via SD coincides with the set of PE matchings (see, for example, Theorem 2 in \cite{IK24}).

Given an instance $(\studentset,\schoolset,\succ_{\studentset},(\student,\school))$ of the SD feasibility problem, 
we construct a matching market with student acquaintance graph $\instance=(\tilde{\studentset}, \tilde{\schoolset}, X, \succ_{\tilde{\studentset}}, \succ_{\tilde{\schoolset}}, q, G)$ as follows.
\begin{itemize}
\item $\studentset = \studentset \cup \{ \tilde{\student} \}$, $\tilde{\schoolset} = \schoolset \cup \{ \tilde{\school} \}$, where $\tilde{\student} \notin \studentset$  and $\tilde{\school} \notin \schoolset$, and $X = \tilde{\studentset} \times \tilde{\schoolset}$.
\item The restriction of $\succ_{\tilde{\studentset}}$ on $\studentset$ is $\succ_{\studentset}$.
$\tilde{\school}$ is the least preferred school for all students in $\tilde{\studentset}$, and $\tilde{\student}$'s preferences are otherwise arbitrary.
\item $\succ_{\tilde{\schoolset}}$ is defined as follows:
\begin{align*}
&\school: \student \succ \student_{1} \succ \dots \succ \student_{|I|-1} \succ \tilde{\student} \\
&\tilde{\schoolset} \setminus \{\school\}: \student_{1} \succ \dots \succ \student_{|I|-1} \succ \tilde{\student} \succ \student,
\end{align*}
where we assume that $\studentset \setminus \{\student\} = \{ \student_1, \dots, \student_{|I|-1} \}$.
We note that preferences among the students in $\studentset \setminus \{\student\}$ can actually be arbitrary.
\item $q_{\school} = 1$ for all $\school \in \tilde{\schoolset}$.
\item $G=(\tilde{\studentset},E)$ is defined as $E = \{ \{\student,\tilde{\student}\}  \} \cup \{ \{ \student,\student' \} \mid \student' \in \studentset \setminus \{ \student \} \} \cup \{ \{ \tilde{\student},\student' \} \mid \student' \in \studentset \setminus \{ \student \} \}$ (see \cref{fig:NP-complete-proof}).
In other words, we assume that the students in $\studentset \setminus \{ \student \}$ are not acquainted with each other.
\end{itemize}

We first claim that a matching $\matching$ of size $|\tilde{\studentset}|$ in the matching market $\instance$ is LEF if and only if $(\student,\school) \in \matching$ and $(\tilde{\student},\tilde{\school}) \in \matching$.
For the "if" part, assume that $(\student,\school) \in \matching$ and $(\tilde{\student},\tilde{\school}) \in \matching$.
Then $\student$ is not envied, since $\student$ is the most preferred student by $\school$.
Also, $\tilde{\student}$ is not envied, since $\tilde{\school}$ is the least preferred school for any student.
Finally, any student other than $\student$ and $\tilde{\student}$ is not envied by its neighbors $\student$ and $\tilde{\student}$, since $\student$ is the least preferred student for the schools other than $\school$ and $\tilde{\student}$ is less preferred by the schools any student in $\studentset \setminus \{\student\}$.
For the "only-if" part, assume that $|\tilde{\studentset}|$ is LEF.
It follows that $(\student, \tilde{\school}) \notin \matching$ since otherwise $\student$ has a local envy to the student who is matched to $\school$.
Then we have $(\tilde{\student},\tilde{\school}) \in \matching$, since otherwise the student who is matched to $\tilde{\school}$ has a local envy to $\tilde{\student}$ (as the student is preferred by $\tilde{\school}$ to $\tilde{\student}$).
It further follows that $(\student,\school) \in \matching$, since otherwise $\tilde{\student}$ has a local envy to $\student$.
Therefore, we have $(\student,\school) \in \matching$ and $(\tilde{\student},\tilde{\school}) \in \matching$.

Now, we show that there exists a PE matching in $(\studentset,\schoolset,\succ_{\studentset},(\student,\school))$ that matches $\student$ and $\school$ if and only if there exists an LEF and PE matching in the matching market $\instance$.

First, assume that there exists a PE matching in $(\studentset,\schoolset,\succ_{\studentset},(\student,\school))$ that matches $\student$ and $\school$.
Define a matching $\tilde{\matching}$ of the matching market $\instance$ as $\tilde{\matching} = \matching \cup \{(\tilde{\student},\tilde{\school})\}$.
From the above claim, $\tilde{\matching}$ is LEF.
Moreover, since $\tilde{\school}$ is the least preferred school for all students in $\tilde{\studentset}$ and $\matching$ is PE, $\tilde{\matching}$ is also PE.

Second, assume that there exists a LEF and PE matching $\matching \subseteq X$ in the matching market $\instance$.
From the above claim, we have $(\student,\school) \in \matching$ and $(\tilde{\student},\tilde{\school}) \in \matching$.
Then restricting $\matching$ on $\studentset$ yields a PE matching in $(\studentset,\schoolset,\succ_{\studentset},(\student,\school))$ that matches $\student$ and $\school$.
\end{proof}
\fi

\ifconference
\else
Although the fundamental theorems do not hold in our model as shown above, 
we demonstrate in \cref{sec:new-mechanisms} that efficient matchings can be obtained by restricting graph structures and schools' preferences.
\fi
\section{Impossibility results}

In this section, we clarify the extent to which local envy-freeness (LEF) is incompatible with Pareto efficiency (PE) in general. 
Specifically, we present the following impossibility results:
\begin{itemize}
\item PE and LEF are incompatible on a path graph among three students (\cref{thm:no-LEF-PE-under-path}).
\item strategyproofness (SP), PE and LEF are incompatible on a path graph among three students even if schools' preferences are single-peaked on the graph (\cref{thm:no-SP-LEF-PE-under-path-single-peaked}).
\item PE and LEF are incompatible on a cycle graph among three students even if schools' preferences are single-peaked on the graph (\cref{thm:no-LEF-PE-under-cycle-single-peaked}). (See \cref{subsec:treewidthk-LEF-k-1} for a definition of single-peaked preferences on general graphs.)
\end{itemize}

It is known that no mechanism can be simultaneously PE and fair (i.e., justified-envy-free) in general~\cite{roth1982economics,AS03}, where fairness coincides with LEF with respect to complete graphs.
We first show that no mechanism can be simultaneously PE and LEF even under very simple graphs, i.e., path graphs.

\begin{theorem}
\label{thm:no-LEF-PE-under-path}
No mechanism can be simultaneously PE and LEF with respect to path graphs.
\end{theorem}

\ifconference
\else
\ifconference
\begin{figure}[t]
\centering
\includegraphics[scale = 0.60]{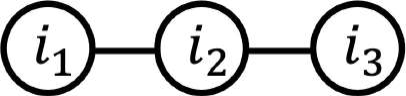}
\hspace{5cm}
\caption{A path graph.}
\label{fig:path}
\hspace{10cm}
\end{figure}

\setcounter{theorem}{\getrefnumber{thm:no-LEF-PE-under-path}}
\addtocounter{theorem}{-1}
\begin{theorem}
No mechanism can be simultaneously PE and LEF with respect to path graphs.
\end{theorem}
\begin{proof}
\else
\begin{proof}
\fi
Let $\studentset = \{\student_1, \student_2, \student_3\}$, $\schoolset = \{\school_1, \school_2, \school_3\}$, and $q_{\school_j} = 1$ for all $j$.
Assume that the student acquaintance graph is a path graph depicted in~\cref{fig:path} in~\cref{sec:failure-fundamental-theorems}.

Assume that $\succ_\student$ and $\succ_\school$ are given as follows.
\begin{align*}
\student_1: \school_1 \succ_{\student_1} \school_2 \succ_{\student_1} \school_3 \qquad \school_1: \student_2 \succ_{\school_1} \student_1 \succ_{\school_1} \student_3\\
\student_2: \school_2 \succ_{\student_2} \school_1 \succ_{\student_2} \school_3 \qquad \school_2: \student_1 \succ_{\school_2} \student_3 \succ_{\school_2} \student_2\\
\student_3: \school_1 \succ_{\student_3} \school_2 \succ_{\student_3} \school_3 \qquad \school_3: \student_1 \succ_{\school_3} \student_2 \succ_{\school_3} \student_3
\end{align*}

We enumerate all Pareto efficient matchings as follows:
\begin{itemize}
\item $Y_1 := \{$$(\student_1, \school_1)$, $(\student_2, \school_2)$, $(\student_3, \school_3)$$\}$
\item $Y_2 := \{$$(\student_1, \school_1)$, $(\student_2, \school_3)$, $(\student_3, \school_2)$$\}$
\item $Y_3 := \{$$(\student_1, \school_2)$, $(\student_2, \school_3)$, $(\student_3, \school_1)$$\}$
\item $Y_4 := \{$$(\student_1, \school_3)$, $(\student_2, \school_2)$, $(\student_3, \school_1)$$\}$
\end{itemize}
We see that all of these matchings do not satisfy LEF.
Indeed, in matching $Y_1$, $\student_3$ has justified envy towards $\student_2$ since $\school_2 \succ_{\student_3} \school_3$ and $\student_3 \succ_{\school_2} \student_2$.
In matching $Y_2$, $\student_2$ has justified envy towards $\student_1$ since $\school_1 \succ_{\student_2} \school_3$ and $\student_2 \succ_{\school_1} \student_1$.
In matching $Y_3$, $\student_2$ has justified envy towards $\student_3$ since $\school_1 \succ_{\student_2} \school_3$ and $\student_2 \succ_{\school_1} \student_3$.
In matching $Y_4$, $\student_1$ has justified envy towards $\student_2$ since $\school_2 \succ_{\student_1} \school_3$ and $\student_1 \succ_{\school_2} \student_2$.
Therefore, there exists no PE and LEF matching for the above instance.
\end{proof}

In the proof of Theorem 5, we show that even in a path graph of length three, PE and LEF are incompatible. 
This result can be extended to any graph that contains a path of length three as an induced subgraph, where the same incompatibility holds.
\fi

We next show that even if the schools' preferences are single-peaked on paths no mechanism can be simultaneously PE and LEF if we require SP.

\begin{theorem}
\label{thm:no-SP-LEF-PE-under-path-single-peaked}
No mechanism can be simultaneously SP, PE, and LEF with respect to path graphs even if the schools' preferences are single-peaked on the graphs.
\end{theorem}

\ifconference
\else
\begin{table}[!t]
\caption{Possible matchings for preference profiles (Theorem~\ref{thm:no-SP-LEF-PE-under-path-single-peaked})}
\label{table:no-SP-LEF-PE-path}
\centering
\begin{tabular}{lllll}
\hline preference & $\student_1$  & $\student_2$ & $\student_3$ & possible  \\
profile & & & & matchings \\
\hline \hline
$\succ^1_\studentset$  & $\school_1 \school_2 \school_3$  & $\school_1 \school_2 \school_3$ & $\school_2 \school_1 \school_3$ & [$\school_3, \school_1, \school_2$]\\ 
\hline
$\succ^2_\studentset$  & $\school_2 \school_1 \school_3$  & $\school_1 \school_2 \school_3$ & $\school_2 \school_1 \school_3$ & [$\school_3, \school_1, \school_2$]\\ 
&&&&[$\school_2, \school_3, \school_1$] \\
\hline
$\succ^3_\studentset$  & $\school_2 \school_1 \school_3$  & $\school_2 \school_1 \school_3$ & $\school_2 \school_1 \school_3$ & [$\school_2, \school_3, \school_1$]\\  \hline
\end{tabular}
\end{table}
\ifconference
\setcounter{theorem}{\getrefnumber{thm:no-SP-LEF-PE-under-path-single-peaked}}
\addtocounter{theorem}{-1}
\begin{theorem}
No mechanism can be simultaneously SP, PE, and LEF with respect to path graphs even if the schools' preferences are single-peaked on the graphs.
\end{theorem}
\begin{proof}
\else
\begin{proof}
\fi
Consider a matching market with three 
students $\studentset = \{\student_1, \student_2, \student_3\}$ and three 
schools $\schoolset = \{\school_1, \school_2, \school_3\}$.
Let $X = \studentset \times \schoolset$ and $q_{\school_j} = 1$ for all $j \in [3]$.
The schools' preference profile $\succ_\schoolset$ is as follows:
\[
\begin{array}{ccc}
\school_1: \student_3 \succ_{\school_1} \student_2 \succ_{\school_1} \student_1\\
\school_2: \student_1 \succ_{\school_2} \student_2 \succ_{\school_2} \student_3\\
\school_3: \student_1 \succ_{\school_3} \student_2 \succ_{\school_3} \student_3\\
\end{array}
\]

To make the description concise, we denote a preference of students by a sequence of schools. 
For example, we denote $\school_2 \succ_\student \school_3 \succ_\student \school_1$ as $\school_2 \school_3 \school_1$.
Furthermore, we denote a matching as a sequence of schools assigned to $\student_1$, $\student_2$ and $\student_3$.
For example, we denote matching $\{(\student_1,\school_3), (\student_2,\school_1), (\student_3,\school_2)\}$ as $[\school_3, \school_1, \school_2]$.

Assume, for the sake of contradiction, that there exists a SP, PE, and LEF mechanism. 
We examine three students' preference profiles: $\succ^1_\studentset, \succ^2_\studentset$, and $\succ^3_\studentset$. 
These preference profiles and possible matchings that satisfy PE and LEF are summarized in Table~\ref{table:no-SP-LEF-PE-path}.
First, for $\succ_\studentset^1 = (\school_1 \school_2 \school_3, \school_1 \school_2 \school_3, \school_2 \school_1 \school_3)$, by enumerating all the PE matchings, we can see that the only PE and LEF matching is $[\school_3, \school_1, \school_2]$.
Hence, the mechanism must choose $[\school_3, \school_1, \school_2]$.

Next, for $\succ_\studentset^2 = (\school_2 \school_1 \school_3, \school_1 \school_2 \school_3, \school_2 \school_1 \school_3)$, by enumerating all the PE matchings, we can see that the only PE and  LEF matchings are $[\school_3, \school_1, \school_2]$ and $[\school_2, \school_3, \school_1]$.
Then, the mechanism must choose either $[\school_3, \school_1, \school_2]$ or $[\school_2, \school_3, \school_1]$.

Finally, for $\succ_\studentset^3 = (\school_2 \school_1 \school_3, \school_2 \school_1 \school_3, \school_2 \school_1 \school_3)$, by enumerating all the PE matchings, we can see that the only PE and LEF matching is $[\school_2, \school_3, \school_1]$.
Hence, the mechanism must choose $[\school_2, \school_3, \school_1]$.

Assume that the mechanism chooses $[\school_3, \school_1, \school_2]$ for $\succ_\studentset^2$.
Then $\student_2$ has an incentive to manipulate (to modify profile $\succ_\studentset^3$ to $\succ_\studentset^2$) so that she is assigned to $\school_1$. 
Assume otherwise that the mechanism choose $[\school_2, \school_3, \school_1]$ for $\succ_\studentset^2$.
Then $\student_1$ has an incentive to manipulate (to modify profile $\succ_\studentset^1$ to $\succ_\studentset^2$) so that she is assigned to $\school_2$.
This fact violates our assumption that the mechanism is SP. 
\end{proof}
\fi

Finally, we show the following impossibility result.
The definition of a single-peaked preference on general graphs will be provided in \cref{subsec:treewidthk-LEF-k-1}, and at this moment it is enough to know that for a cycle graph of length three any preference is single-peaked on the graph.

\begin{theorem}
\label{thm:no-LEF-PE-under-cycle-single-peaked}
No mechanism can be simultaneously PE and LEF with respect to cycle graphs even if the schools' preferences are single-peaked on the graphs.
\end{theorem}

\ifconference
\else
\ifconference
\setcounter{theorem}{\getrefnumber{thm:no-LEF-PE-under-cycle-single-peaked}}
\addtocounter{theorem}{-1}
\begin{theorem}
No mechanism can be simultaneously PE and LEF with respect to cycle graphs even if the schools' preferences are single-peaked on the graphs.
\end{theorem}
\begin{proof}
\else
\begin{proof}
\fi
Since any preference profile is single-peaked on a cycle graph of length three and LEF coincides with fairness in the (complete) graph, the statement follows from the impossibility result that PE and fairness are incompatible on a matching market with three students and three schools (see the example in Section 6 of~\cite{roth1982economics} and Example 1 in~\cite{AS03}).
\end{proof}
\fi
\section{New mechanisms}
\label{sec:new-mechanisms}
In this section, we propose several mechanisms achieving Pareto efficiency (PE) and some levels of local envy-freeness (LEF) by restricting the student acquaintance graph and schools' preferences.

\ifconference
\else
We first show that if the graph is a tree and schools' preferences are single-peaked on the tree, then there exists a mechanism that achieves PE, LEF, and mutually-best (MB) in Subsection~\ref{subsec:PE-LEF-trees}.

We then extend the above result to graphs that are ``tree-like'' in two ways.

First, we consider a graph with an underlying tree to which edges are added such that the degree (that is, the number of neighboring vertices) of each vertex is at most $k$ with an assumption that the schools' preferences are single-peaked on the underlying tree.
For this setting, there might be no PE and LEF matching, but we show that there exists a mechanism that finds a matching that satisfies PE, locally EF-($k-1$), and locally ERF-($k-1$) in Subsection~\ref{subsec:PE-LEF-trees}.

Second, we consider graphs with treewidth $k$, where $k$ is a positive integer.
Treewidth is a graph-theoretic parameter that quantifies how similar a given graph is to a tree (see \cref{def:treewidth} in \cref{subsec:treewidthk-LEF-k-1} for a formal definition).
In fact, a graph is a tree if and only if it has treewidth one and connected.
We extend the single-peaked preferences on trees to those on tree-width $k$ graphs and propose a mechanism that satisfies PE, locally ERF-($k-1$), and MB when the graph has treewidth $k$ and the schools' preferences are single-peaked on the graph (\cref{thm:LEF-k-1-PE-under-k-degenerate-single-peaked}) in Subsection~\ref{subsec:treewidthk-LEF-k-1}.

Finally, we show that if the graph has treewidth $k$, then there exists a mechanism that achieves SP, PE, and local ERF-$k$ or local EF-$k$ in Subsection~\ref{subsec:SP-PE-LEF-k}.
As can be seen from the impossibility theorem presented in the previous section, our mechanisms achieve the best possible fairness guarantees in terms of the number of students by whom each student is envied.
See also \cref{tbl:results} in \cref{sec:intro} for a summary of the results.
\fi 

\subsection{LEF and PE mechanism for tree graphs under single-peaked preferences}
\label{subsec:PE-LEF-trees}

Here, we assume that the student acquaintance graph $G$ is a tree and schools' preferences are single-peaked on $G$, and 
propose a mechanism called a \emph{best to locally top 2 (B-LT2)} mechanism, which satisfies PE, LEF, and MB.

\begin{algorithm}[t]
\begin{algorithmic}[1]
\REQUIRE a matching market $\instance=(\studentset, \schoolset, X, \succ_\studentset, \succ_\schoolset, q, G)$
\ENSURE a matching $Y$
\caption{Best to locally top 2 (B-LT2) mechanism}
\label{alg:top-2}
\STATE $Y \leftarrow \emptyset$.
\STATE For each MB-pair $(\student,\school)$, assign $\student$ to $\school$ (i.e., $Y \leftarrow Y \cup \{ (\student,\school \}$).
\STATE 
Choose an arbitrary unassigned student $\student$.
If there are no schools with remaining seats (i.e., schools $\school$ such that $|Y_{\school}| < q_{\school}$) that are acceptable to $\student$, $\student$ is not assigned to any school and is regarded as assigned to an emptyset.
Otherwise, $\student$ chooses the most preferable school $\school$ among the acceptable schools with remaining seats, and 
if $\school$ prefers $\student$ to all the unassigned neighbor students of $\student$, then assign $\student$ to $\school$.

Repeat this step until no assignment is made.
\STATE 
Each unassigned student $\student$ chooses the most preferable school $\school$ among the acceptable schools with a remaining seat.
$\school$ must prefer the unique unassigned student $\student'$ in $N(\student)$ to $\student$, and we say $\student'$ attacks $\student$.
Find an arbitrary pair of unassigned students attacking each other and 
assign them to their (different) most preferable schools.
\STATE If all students are assigned, then let $Y$ be the final matching and terminate the mechanism.
Otherwise, go to 3.
\end{algorithmic}
\end{algorithm}

B-LT2 mechanism is described in \cref{alg:top-2}.
First, it matches the MB-pairs (Step 2).
Second, it repeatedly assigns a student to her most preferable school (among the schools with a remaining seat) if she is most preferable among her neighboring students for the school until no such student remains (Step 3).
Third, each student $\student'$ is said to attack another student $\student$ if $\student'$ is a neighbor of $\student$ and $\student$'s most preferable school prefers $\student'$ to $\student$, and the mechanism chooses and assigns a pair of students who attack each other to their most preferable schools (Step 4).

Note that the selection of a student in Steps 3 and 4 can be determined according to the implementation policy, which may involve random selection, a master-list based on school preferences, or other methods.
Regardless of the method chosen, the mechanism is theoretically guaranteed to satisfy PE and LEF.

We first show an auxiliary lemma.

\begin{lemma}
\label{lem:single-peaked-locally-top-2}
Let $\succ_\school$ be a preference over the set $\studentset$ of students that is single-peaked on a tree $G$.
Then each student $\student$ is ranked within the top two in the set consisting of the student and her neighbors.
Namely, either (a) $\student \succ_\school \student'$ for every $\student' \in N(\student)$ or (b) $\student' \succ_\school \student$ for some $\student' \in N(\student)$ and $\student \succ_\school \student''$ for every $\student'' \in N(\student)\setminus \{\student'\}$.
\end{lemma}

\begin{proof}
For the sake of contradiction, assume that $\student_1 \succ_\school \student_2 \succ_\school \student$ for $\student_1, \student_2 \in N(\student)$.
Then, the set of students that is more preferable or equal to $\student_2$ is not connected in $G$, which contradicts that $\succ_\school$ is single-peaked on $G$.
\end{proof}

Now, we show that B-LT2 satisfies several desirable properties.

\begin{theorem}
\label{thm:LEF-PE-under-tree-single-peaked}
B-LT2 satisfies PE, LEF, and MB, 
assuming that the student acquaintance graph $G$ is a tree and schools' preferences are single-peaked on the tree.
\end{theorem}

\ifconference
\begin{proof}
We provide only a proof sketch. See the appendix for details.

First, we show that B-LT2 stops in finite steps.
It suffices to show that in Step 4 of B-LT2 some students are assigned to some schools.
After Step 3, each unassigned student is in the second place in the set consisting of the student and her neighbors for the school she prefers most by \cref{lem:single-peaked-locally-top-2}.
Hence, each student attacks some other neighboring student.
Thus, there exists a cycle $\student_{j_1}, \student_{j_2}, \dots, \student_{j_\ell}, \student_{j_1}$ in $G$ such that $\student_{j_k}$ attacks $\student_{j_{k+1}}$ for $k \in \{ 1,\dots, \ell -1 \}$ and $\student_{j_\ell}$ attacks $\student_{j_1}$.
Since $G$ is a tree, the length $\ell$ of this cycle must be two.
This means that $\student_{j_1}$ and $\student_{j_2}$ attack each other, and these students are assigned to some schools in Step 4.
Therefore, eventually all students are assigned to some schools and B-LT2 stops in finite steps.

The matching $\matching$ output by B-LT2 clearly satisfies MB, and  is PE, since the output of B-LT2 coincides with the output obtained by an SD.

Finally, we show that $Y$ is LEF.
We show this by showing that each student cannot receive justified envies from the neighbors.
If a student is assigned to a school in step 2 of B-LT2, then clearly she receives no justified envies from the neighbors.
If a student $\student$ is assigned to a school $\school$ in step 3, then $\student$ also receives no justified envies from the neighbors since each of the neighbors of $\student$ is either (a) already assigned to a school more preferable to $\school$ or (b) less preferable to $\student$ for $\school$.
Finally, 
if a student $\student$ is assigned to a school $\school$ in step 4, then 
each of the neighbors of $\student$ is either (a) already assigned to a school more preferable to $\school$, (b) less preferable to $\student$ for $\school$, or (c) simultaneously assigned to a school more preferable to $\school$.
Hence, $\student$ does not receive justified envies from the neighbors.
It follows that $Y$ is LEF.
\end{proof}
\else
\ifconference
\setcounter{theorem}{\getrefnumber{thm:LEF-PE-under-tree-single-peaked}}
\addtocounter{theorem}{-1}
\begin{theorem}
B-LT2 satisfies PE, LEF, and MB, 
assuming that the student acquaintance graph $G$ is a tree and schools' preferences are single peaked on the tree.
\end{theorem}
\begin{proof}
\else
\begin{proof}
\fi
First, we show that B-LT2 stops in finite steps.
It suffices to show that in Step 4 of B-LT2 some students are assigned to some schools.
After Step 3, each remaining (i.e., unmatched) student is in the second place among the set consisting of the student and her neighbors for the school she prefers most by \cref{lem:single-peaked-locally-top-2}.
Hence, each student attacks some other neighboring student.
Thus, there exists a cycle $\student_{j_1}, \student_{j_2}, \dots, \student_{j_\ell}, \student_{j_1}$ in $G$ such that $\student_{j_k}$ attacks $\student_{j_{k+1}}$ for $k \in \{ 1,\dots, \ell -1 \}$ and $\student_{j_\ell}$ attacks $\student_{j_1}$.
Since $G$ is a tree, the length $\ell$ of this cycle must be two.
This means that $\student_{j_1}$ and $\student_{j_2}$ attack each other, and these students are assigned to some schools in Step 4.
Therefore, eventually all students are assigned to some schools and B-LT2 stops in finite steps.

Second, B-LT2 clearly outputs a feasible matching, since it always assigns a student to a school with a remaining seat.

Third, B-LT2 clearly satisfies MB, since it matches MB-pairs in step 2.

Fourth, we show that the matching $\matching$ output by B-LT2 is PE.
This is because, 
each student is assigned to the most preferable school among the remaining schools when they are assigned (or no acceptable schools have remaining seats). 
Thus, the output of B-LT2 coincides with the output obtained by SD. 
Since the output of SD is PE, the output of B-LT2 is also PE.

Finally, we show that $Y$ is LEF.
We show this by showing that each student cannot receive justified envies from the neighbors.
If a student is assigned to a school in step 2 of B-LT2, then clearly she receives no justified envies from the neighbors.
If a student $\student$ is assigned to an emptyset in step 3, then she clearly receives no justified envy.
If a student $\student$ is assigned to a school $\school$ in step 3, then $\student$ also receives no justified envies from the neighbors since each of the neighbors of $\student$ is either (a) already assigned to a school more preferable to $\school$ or (b) less preferable to $\student$ for $\school$.
Finally, 
if a student $\student$ is assigned to a school $\school$ in step 4, then 
each of the neighbors of $\student$ is either (a) already assigned to a school more preferable to $\school$, (b) less preferable to $\student$ for $\school$, or (c) simultaneously assigned to a school more preferable to $\school$.
Hence, $\student$ does not receive justified envies from the neighbors.
It follows that $Y$ is LEF.
\end{proof}

\fi

We note that B-LT2 does not satisfy SP.
This follows from the impossibility theorem (\cref{thm:no-SP-LEF-PE-under-path-single-peaked}), 
since a path graph is a tree graph.

We can generalize \cref{thm:LEF-PE-under-tree-single-peaked} as follows.

\begin{theorem}\label{thm:tree+several-edges}
Consider a graph with an underlying tree to which edges are added such that the degree of each vertex is at most $k$, and assume that schools' preferences are single-peaked on the underlying tree.
For this setting, there exists a mechanism that finds a matching that satisfies PE, local EF-($k-1$), and local ERF-($k-1$).
\end{theorem}

\ifconference
\else
\ifconference
\setcounter{theorem}{\getrefnumber{thm:tree+several-edges}}
\addtocounter{theorem}{-1}
\begin{theorem}
Consider a graph with an underlying tree to which edges are added such that the degree of each vertex is at most $k$, and assume that schools' preferences are single-peaked on the underlying tree.
For this setting, there exists a mechanism that finds a matching that satisfies PE, locally EF-($k-1$), and locally EDF-($k-1$).
\end{theorem}
\begin{proof}
\else
\begin{proof}
\fi
It suffices to run B-LT2 to the underlying tree and output the obtained matching.
This matching is clearly PE.
Moreover, it is LEF with respect to the underlying tree and each student has at most $k-1$ neighbors other than those connected by this tree, implying that it is locally EF-($k-1$) and locally EDF-($k-1$).
\end{proof}
\fi

Graphs that satisfy the assumptions in the above theorem appear as a road network with several bypasses to the road network on a tree.

Finally, we compare matchings obtained by the B-LT2 and the deferred acceptance (DA) mechanisms.

\begin{theorem}\label{thm:comparison-BLT2-and-DA}
There exists a matching market with a student acquaintance graph in which any of the possible output matchings of the B-LT2 mechanism does not Pareto dominate the student-optimal stable matching (that is, the output of the DA mechanism).
\end{theorem}

\ifconference
\else
\ifconference
\setcounter{theorem}{\getrefnumber{thm:comparison-BLT2-and-DA}}
\addtocounter{theorem}{-1}
\begin{theorem}
There exists a matching market with a student acquaintance graph in which any of the possible output matchings of the B-LT2 mechanism does not Pareto dominate the student-optimal stable matching (that is, the output of the DA mechanism).
\end{theorem}
\begin{proof}
\else
\begin{proof}
\fi
Let $\studentset = \{\student_1, \student_2, \student_3, \student_4, \student_5\}$, $\schoolset = \{\school_1, \school_2, \school_3, \school_4, \school_5\}$, $X = \studentset \times \schoolset$, and $q_{\school_j} = 1$ for all $j \in [5]$.
Assume that the student acquaintance graph is a path graph with edge set $\{ \{\student_1,\student_2\},\{\student_2,\student_3\},\{\student_3,\student_4\},\{\student_4,\student_5\} \}$.

Assume that $\succ_\student$ and $\succ_\school$ are given as follows.
\begin{align*}
\student_1: \school_2 \succ_{\student_1} \school_3 \succ_{\student_1} \school_4 \succ_{\student_1} \school_1 \succ_{\student_1} \school_5\\
\student_2: \school_1 \succ_{\student_2} \school_2 \succ_{\student_2} \school_3 \succ_{\student_2} \school_4 \succ_{\student_2} \school_5\\
\student_3: \school_4 \succ_{\student_3} \school_3 \succ_{\student_3} \school_2 \succ_{\student_3} \school_1 \succ_{\student_3} \school_5\\
\student_4: \school_2 \succ_{\student_4} \school_3 \succ_{\student_4} \school_4 \succ_{\student_4} \school_1 \succ_{\student_4} \school_5\\
\student_5: \school_3 \succ_{\student_5} \school_5 \succ_{\student_5} \school_1 \succ_{\student_5} \school_2 \succ_{\student_5} \school_4\\
\school_1: \student_2 \succ_{\school_1} \student_1 \succ_{\school_1} \student_3 \succ_{\school_1} \student_4 \succ_{\school_1} \student_5\\
\school_2: \student_5 \succ_{\school_2} \student_4 \succ_{\school_2} \student_3 \succ_{\school_2} \student_2 \succ_{\school_2} \student_1\\
\school_3: \student_1 \succ_{\school_3} \student_2 \succ_{\school_3} \student_3 \succ_{\school_3} \student_4 \succ_{\school_3} \student_5\\
\school_4: \student_4 \succ_{\school_4} \student_3 \succ_{\school_4} \student_2 \succ_{\school_4} \student_1 \succ_{\school_4} \student_5\\
\school_5: \student_1 \succ_{\school_5} \student_2 \succ_{\school_5} \student_3 \succ_{\school_5} \student_4 \succ_{\school_5} \student_5
\end{align*}
Let $Y_1$ denote the matching obtained by the DA mechanism:
\begin{align*}
Y_1 = \{(\student_1, \school_3), (\student_2, \school_1), (\student_3, \school_4), (\student_4, \school_2), (\student_5, \school_5)\}.
\end{align*}
Let $Y_2$ be the only matching that can be obtained by the B-LT2 mechanism:
\begin{align*}
Y_2 = \{(\student_1, \school_2), (\student_2, \school_1), (\student_3, \school_4), (\student_4, \school_3), (\student_5, \school_5)\}.
\end{align*}
It can be confirmed that $Y_1$ and $Y_2$ are both PE matchings. 
This implies that although the output of B-LT2 is PE, it does not necessarily Pareto dominate the matching produced by the DA mechanism.
\end{proof}

\fi

\subsection{
Locally ERF-($k-1$) and PE mechanism for treewidth-$k$ graphs under single-peaked preferences
}
\label{subsec:treewidthk-LEF-k-1}

In this subsection we extend the result for trees in the previous subsection to graphs that are ``close to'' trees, i.e., graphs with treewidth $k$ for some positive integer $k$.
Treewidth is a concept frequently used in the field of algorithms to generalize algorithms for trees. For textbooks on treewidth, for example, \cite{CFK+15} is available.

We first provide a definition of the treewidth of a graph.

\begin{definition}[treewidth]\label{def:treewidth}
A tree decomposition of a graph $G = (\studentset,E)$ is a tree $T$ with nodes $B_1,\dots, B_\ell$ (sometimes called \emph{bags}), where each $B_j$ is a subset of $\studentset$, satisfying the following properties:
\begin{itemize}
\item $\studentset = \bigcup_{1 \le j \le \ell} B_j$.
\item For each $\student \in \studentset$, the tree nodes containing $\student$ form a connected subtree of $T$.
\item For each edge $\{ \student,\student' \} \in E$ there exists $j \in [\ell]$ such that $\{ \student,\student' \} \subseteq B_j$.
\end{itemize}
The \emph{width} of a tree decomposition is the size of its largest bag $B_j$ minus one.
The \emph{treewidth} of $G$ is the minimum width among all the tree decompositions of $G$.
\end{definition}

A graph is tree if and only if it is connected and has treewidth one.
Any graph $G =(\studentset,E)$ has treewidth at most $n-1$, since there exists a trivial tree decomposition $T$ such that $T$ has a unique node $B_1=\studentset$.
For a graph drawn in \cref{fig:treewidth2}, its tree decomposition of width 2 is drawn in \cref{fig:treewidth3}.

\begin{figure}[t]
\centering
\includegraphics[scale = 0.35]{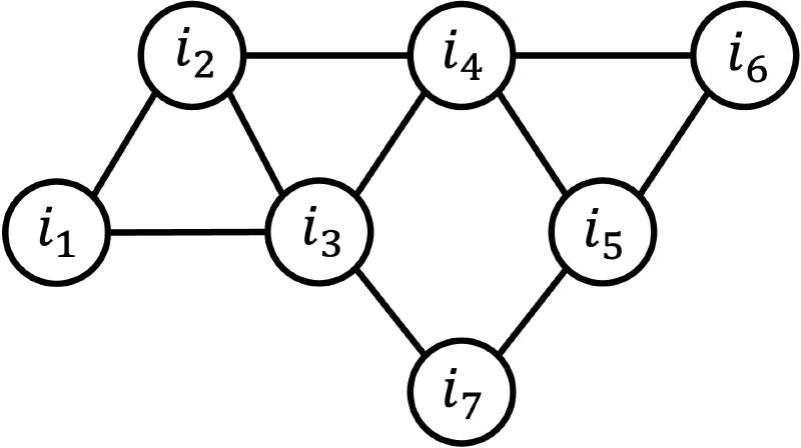}
\hspace{5cm}
\caption{A graph that has treewidth 2.}
\label{fig:treewidth2}
\hspace{10cm}
\end{figure}

\begin{figure}[t]
\centering
\includegraphics[scale = 0.2]{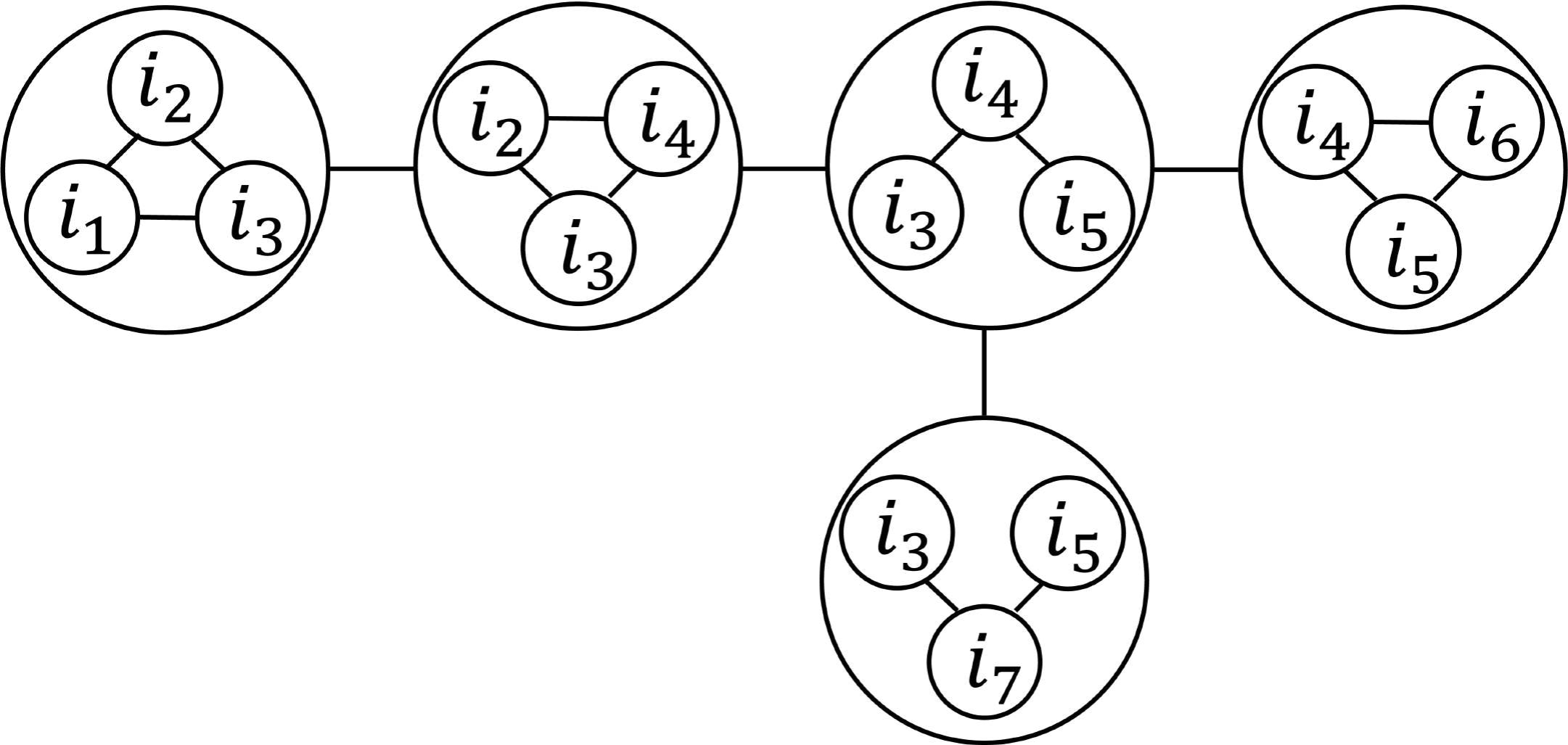}
\hspace{5cm}
\caption{A tree decomposition of the graph in \cref{fig:treewidth2} that has width 2.}
\label{fig:treewidth3}
\hspace{10cm}
\end{figure}

We will use the following well-known result in graph theory.\footnote{This fact is frequently mentioned in papers in the field of algorithms, but we were unable to find any literature that provides a proof. See also Exercise 7.14 in \cite{CFK+15}.}
\begin{lemma}
\label{lem:treewidth-k->k-degenerate}
Any graph $G$ with treewidth $k$ has a vertex ordering such that for every vertex $\student$ at most $k$ neighbors come after $\student$ in the order.
\end{lemma}

Now, we define single-peaked preferences on a tree decomposition of graph $G$.

\begin{definition}[single-peaked preferences on a tree decomposition]
\label{def:single-peaked-tree-decomposition}
A preference $\succ_\school$ over the set $\studentset$ of students is called \emph{single-peaked on a tree decomposition $T$ of $G =(\studentset,E)$} if it is constructed in the following way:
\begin{enumerate}
\item Choose any node $B_j$ of $T$. 
Rank the students in $B_j$ from top.
\item Choose any node $B_{j'}$ that is adjacent to any of the previously chosen nodes.
Rank the non-ranked students in $B_{j'}$ following the students ranked so far.
\item If every student is ranked, then stop.
Otherwise, go to 2.
\end{enumerate}
\end{definition}

\begin{example}
Let us define a single-peaked preference on the tree decomposition in \cref{fig:treewidth3}.
First, we choose the node $\{ \student_3, \student_4, \student_5 \}$ and rank the students in it as $\student_4 \succ \student_3 \succ \student_5$.
Second, we choose the node $\{ \student_2, \student_3, \student_4 \}$ adjacent to the first node and rank $\student_2$ fourth.
Third, we choose the node $\{ \student_1, \student_2, \student_3 \}$ adjacent to the second node and rank $\student_1$ fifth.
Fourth, we choose the node $\{ \student_3, \student_5, \student_7 \}$ adjacent to the first node and rank $\student_7$ sixth.
Finally, we choose the node $\{ \student_4, \student_5, \student_6 \}$ adjacent to the first node and rank $\student_6$ seventh.
In this way, we obtain a single-peaked preference 
\begin{align*}
\student_4 \succ \student_3 \succ \student_5 \succ \student_2 \succ \student_1 \succ \student_7 \succ \student_6
\end{align*}
on the tree decomposition in \cref{fig:treewidth3}.
\end{example}

As seen from the above example, a single-peaked preference on a tree decomposition arises when we represent a road network by a tree-like structure via a tree decomposition, 
although the distance information is abstracted in graphs.
Namely, the first chosen node $\{ \student_3, \student_4, \student_5 \}$ consists of closest students to a school, the second node $\{ \student_2, \student_3, \student_4 \}$ consists of second closest students to the school, and so on.
It has been demonstrated by \citet{MSJ19} that treewidths of actual road networks are low.

Now, we show an auxiliary lemma that extends \cref{lem:single-peaked-locally-top-2} for tree graphs.

\begin{lemma}
\label{lem:single-peaked-locally-top-k+1}
Let $\succ_\school$ be a preference over the set $\studentset$ of students that is single-peaked on a tree decomposition $T$ of width $k$.
Then each student $\student$ is ranked within the top $k+1$ in the set consisting of the student and her neighbors.
\end{lemma}

\ifconference
\else
\ifconference
\setcounter{theorem}{\getrefnumber{lem:single-peaked-locally-top-k+1}}
\addtocounter{theorem}{-1}
\begin{lemma}
Let $\succ_\school$ be a preference over the set $\studentset$ of students that is single-peaked on a tree decomposition $T$ of width $k$.
Then each student $\student$ is ranked within the top $k+1$ among the set consisting of the student and her neighbors.
\end{lemma}
\begin{proof}
\else
\begin{proof}
\fi
For the sake of contradiction, assume that $\student_1 \succ_\school \student_2 \succ_\school \dots \succ_\school \student_{k+1} \succ_\school \student$ for $\student_1, \dots, \student_{k+1} \in N(\student)$.
When constructing $\succ_\school$ in the way described in \cref{def:single-peaked-tree-decomposition}, let $B_j$ be the first chosen node of $T$ that includes $\student_1$, and 
let $B_{j'}$ be the first chosen node of $T$ that includes $\student_{k+1}$.
From the construction there exists a path from $B_j$ to $B_{j'}$ in $T$, where all the intermediate nodes in the path do not contain $\student$ since $\student$ is ranked lower than $\student_{k+1}$.
Also, since $\{\student_1,\student\},\{ \student_{k+1},\student\} \in E$, there exist nodes $B_{p}$ and $B_{r}$ of $T$ that respectively contain $\{\student_1,\student\}$ and $\{ \student_{k+1},\student\}$.
Since the nodes containing the same vertex form a connected subtree in $T$, 
there exists a path from $B_j$ to $B_p$, from $B_p$ to $B_r$, and from $B_r$ to $B_{j'}$ in $T$.
Concatenating these paths yields a path from $B_j$ to $B_{j'}$ in $T$ in which at least one intermediate node contains $\student$.
Hence, there exist two distinct paths from $B_j$ to $B_{j'}$ in $T$, contradicting that $T$ is a tree.
\end{proof}

\fi

Now, we describe our mechanism called the \emph{best to locally top $k+1$ (B-LT($k+1$))} mechanism (in \cref{alg:Top-k+1}) that extends B-LT2 proposed in the previous subsection.

\begin{algorithm}[t]
\begin{algorithmic}[1]
\REQUIRE a matching market $\instance=(\studentset, \schoolset, X, \succ_\studentset, \succ_\schoolset, q, G)$
\ENSURE a matching $Y$
\caption{Best to locally top $k+1$ (B-LT($k+1$))}
\label{alg:Top-k+1}
\STATE $Y \leftarrow \emptyset$.
\STATE For each MB-pair $(\student,\school)$, assign $\student$ to $\school$ (i.e., $Y \leftarrow Y \cup \{ (\student,\school) \}$).
\STATE Choose an arbitrary unassigned student $\student$.
If there are no schools with remaining seats that are acceptable to $\student$, $\student$ is not assigned to any school and is regarded as assigned to an emptyset.
Otherwise, $\student$ chooses the most preferable school $\school$ among the acceptable schools with a remaining seat, and if student $\student$ is ranked within the top $k$ in the set consisting of the student and her unassigned neighbors for $\school$, then assign $\student$ to $\school$.

Repeat this step until no assignment is made.
\STATE Each unassigned student $\student$ chooses the most preferable school $\school$ among the acceptable schools with a remaining seat.
$\school$ must prefer exactly $k$ unassigned students in $N(\student)$ to $\student$, and we say those students attack $\student$.
Choose an arbitrary pair of unassigned students attacking each other and assign these students to their (different) most preferable schools.
\STATE If all students are assigned, then let $Y$ be the final matching and terminate the mechanism.
Otherwise, go to 3.
\end{algorithmic}
\end{algorithm}

\ifconference
\else
B-LT($k+1$) mechanism is described in \cref{alg:Top-k+1}.
First, it  matches MB-pairs.
Second, it repeatedly assigns a student to her most preferable school (among the schools with a remaining seat) if 
she is ranked within the \emph{top $k$} among the set consisting of her and her neighbors for $\school$,
Third, each student $\student'$ is said to attack another student $\student$ if $\student'$ is a neighbor of $\student$ and $\student$'s most preferable school prefers $\student'$ to $\student$, and the mechanism assigns students who attacks each other to their most preferable schools.
\fi 

Now, we show that B-LT($k+1$) mechanism satisfies several desirable properties.

\begin{theorem}
\label{thm:LEF-k-1-PE-under-k-degenerate-single-peaked}
B-LT($k+1$) satisfies PE, local ERF-($k-1$), and MB, 
assuming that the student acquaintance graph $G$ has treewidth $k$ and schools' preferences are single-peaked on a tree 
decomposition of $G$ with width $k$.
\end{theorem}

\ifconference
\else
\ifconference
\setcounter{theorem}{\getrefnumber{thm:LEF-k-1-PE-under-k-degenerate-single-peaked}}
\addtocounter{theorem}{-1}
\begin{theorem}
B-LT($k+1$) satisfies PE, local ERF-($k-1$), and MB, 
assuming that the student acquaintance graph $G$ has treewidth $k$ and schools' preferences are single peaked on a tree 
decomposition of $G$ with width $k$.
\end{theorem}
\begin{proof}
\else
\begin{proof}
\fi
First, we show that B-LT($k+1$) stops in finite steps.
It suffices to show that in Step 4 of B-LT($k+1$) some students are assigned to some schools.
After Step 3, each unassigned student is 
ranked the top $k+1$ among the set consisting of her and her neighbors for the school she most prefers by \cref{lem:single-peaked-locally-top-k+1}.
Hence, each student is attacked by $k$ other neighboring students.
From \cref{lem:treewidth-k->k-degenerate} we can order the vertices of $G$ so that every vertex $\student$ has at most $k$ neighbors that comes after $\student$.
Let $\student$ be the first vertex (student) in this order satisfying that there exists a student $\student'$ such that $\student'$ attacks $\student$ and $\student'$ is before $\student$ in this order.
Since $\student'$ is attacked by exactly $k$ students after $\student'$ in this order and there exists at most $k$ neighbor students after $\student'$, $\student'$ must be attacked by $\student$.
Hence, $\student$ and $\student'$ attack each other and the mechanism assigns some students to some schools in Step 4.
Therefore, eventually all students are assigned to some schools and B-LT($k+1$) stops in finite steps.

Second, B-LT($k+1$) clearly outputs a feasible matching, since it always assigns a student to a school with a remaining seat.

Third, B-LT($k+1$) clearly satisfies MB, since it matches MB-pairs in step 2.

Fourth, we can show that the matching $Y$ output by B-LT($k+1$) is Pareto efficient in the same way as the proof in \cref{thm:LEF-PE-under-tree-single-peaked}.

Finally, we show that $Y$ is locally ERF-($k-1$).
We show this by showing that each student can receive justified envies from at most $k-1$ neighbors.
If a student is assigned to a school in step 2 of B-LT($k+1$), then clearly she receives no justified envies from the neighbors.
If a student $\student$ is assigned to an emptyset in step 3, then she clearly receives no justified envy.
If a student $\student$ is assigned to a school $\school$ in step 3, then $\student$ receives at most $k-1$ justified envies from the neighbors since there remains at most $k-1$ neighbor students that are more preferable to $\student$ for $\school$.
Finally, 
if a student $\student$ is assigned to a school $\school$ in step 4, then 
there remains exactly $k$ neighbor students that are more preferable to $\student$ for $\school$.
However, one of the $k$ neighbor students, say $\student'$, is simultaneously assigned to a school more preferable to $\school$.
Hence, $\student'$ does not embrace local envy toward $\student$.
Therefore, $\student$ receives at most $k-1$ justified envies from the neighbors.
It follows that $Y$ is locally ERF-($k-1$).
\end{proof}

\fi

\subsection{
Locally ERF-$k$ and EF-$k$ mechanisms with SP and PE for treewidth-$k$ graphs under general preferences
}
 \label{subsec:SP-PE-LEF-k}

Here, we propose two SD-based mechanisms that satisfies SP, PE, and locally ERF-$k$ (resp., locally EF-$k$) assuming that the student acquaintance graph $G$ has treewidth $k$.
Unlike the previous two subsections, we do not restrict schools' preference here.

Let us define a master-list $L_d$ as the order described in \cref{lem:treewidth-k->k-degenerate}.
Then we can show that SD using master-list $L_d$ outputs a matching that satisfies SP, PE, and local ERF-$k$.

\begin{theorem}
\label{thm:LreverseEFk-PE-under-treewidth-k}
Assume that the student acquaintance graph $G$ has treewidth $k$.
Then SD using master-list $L_d$ is SP, PE, and locally ERF-$k$.
\end{theorem}

\ifconference
\else
\ifconference
\setcounter{theorem}{\getrefnumber{thm:LreverseEFk-PE-under-treewidth-k}}
\addtocounter{theorem}{-1}
\begin{theorem}
Assume that the student acquaintance graph $G$ has treewidth $k$.
Then SD using master-list $L_d$ is SP, PE, and locally EDF-$k$.
\end{theorem}
\begin{proof}
\else
\begin{proof}
\fi
As it is an SD, it is SP and PE.
Let $Y$ be the resulting matching under the SD.
Observe that when a student $\student$ is assigned to a school in the SD, 
there remain at most $k$ neighbor students that are adjacent to $\student$ in $G$.
Hence, the number of local envies that $\student$ receives in $Y$ is at most $k$.
\end{proof}
\fi

By using a reverse order of $L_d$, we can bound the number of the number of students to whom each student has local envy.
Let $L_d^{\rm rev}$ be the master-list obtained by reversing the order of $L_d$.

\begin{theorem}
\label{thm:LEFk-PE-under-treewidth-k}
Assume that the student acquaintance graph $G$ has treewidth $k$.
Then SD using master-list $L_d^{\rm rev}$ is SP, PE, and locally EF-$k$.
\end{theorem}

\ifconference
\else
\ifconference
\setcounter{theorem}{\getrefnumber{thm:LEFk-PE-under-treewidth-k}}
\addtocounter{theorem}{-1}
\begin{theorem}
Assume that the student acquaintance graph $G$ has treewidth $k$.
Then SD using master-list $L_d^{\rm rev}$ is SP, PE, and locally EF-$k$.
\end{theorem}
\begin{proof}
\else
\begin{proof}
\fi
This can be proven similarly to \cref{thm:LreverseEFk-PE-under-treewidth-k}.
\end{proof}
\fi

\section{Discussions and concluding remarks}\label{sec:conclusion}
We have introduced local envy-freeness to many-to-one two-sided matching and investigated the level of local envy-freeness that can be achieved by Pareto efficient mechanisms.
By restricting the structure of the student acquaintance graph and the schools' preferences, we provided a comprehensive picture of the level of local envy-freeness that can be achieved by Pareto efficient mechanisms, summarized in \cref{tbl:results}.
We specifically focused on graph structures that are ``close to'' trees, i.e., graphs with treewidth $k$ for some positive integer $k$.
Note that our result is applicable to any graph, since the treewidth of any graph is at most the number of its vertices minus one.

In the way of proposing our new mechanisms, we use the fact that for graphs with treewidth $k$ there exists a vertex ordering such that for every vertex $\student$ at most $k$ neighbors comes after $\student$ in the order (\cref{lem:treewidth-k->k-degenerate}).
\ifconference
\else
Indeed, graphs that allow such a vertex ordering are exactly the $k$-degenerate graphs, and our mechanisms based on SD work for those graphs.
Moreover, the B-LT($k+1$) mechanism also works for those graphs assuming that the schools' preference are such that each student is ranked within the top $k+1$ among the set consisting of the student and her neighbors.
Let us provide a formal definition of $k$-degenerate graphs.
\begin{definition}[$k$-degenerate graphs]
A graph is \emph{$k$-degenerate} if 
there exists a vertex ordering such that for every vertex $\student$ at most $k$ neighbors come after $\student$ in the order.
The degeneracy of a graph is the smallest $k$ such that the graph is $k$-degenerate.
\end{definition}

For example, the degeneracy of a tree is one.
Hence, the degeneracy of a graph in a sense measures how it is sparse.
Treewidth is a measure of proximity to a tree and it is known that a graph with treewidth $k$ is $k$-degenerate.
Hence, graphs with bounded degeneracy include graphs with bounded treewidth.

Without changing the corresponding proofs for graphs with treewidth $k$, 
we can show the following.

\begin{theorem}
\label{thm:LreverseEFk-1-PE-under-k-degenerate-top-k+1}
Assume that the graph $G$ is $k$-degenerate and schools' preferences are such that 
each student is ranked within the top $k+1$ in the set consisting of the student and her neighbors.
Then B-LT($k+1$) is PE, locally ERF-($k-1$), and MB.
\end{theorem}

\begin{theorem}
\label{thm:LreverseEFk-PE-under-k-degenerate}
Assume that the graph $G$ is $k$-degenerate.
Then SD using the master list $L_d$ is SP, PE, and locally ERF-$k$.
\end{theorem}

\begin{theorem}
\label{thm:LEFk-PE-under-k-degenerate}
Assume that the graph $G$ is $k$-degenerate.
Then SD using the master list $L_d^{\rm rev}$ is SP, PE, and locally EF-$k$.
\end{theorem}

\fi

Here, let us briefly discuss the expressive power of single-peaked preferences on tree decompositions introduced in this paper.
We demonstrate that even ``multiple-peaked'' preferences, which generalize single-peaked preferences, can sometimes be represented as single-peaked preferences on a tree decomposition that has low width.
Please refer to \cref{fig:two-peaked-preferences}. 
In this figure, students form a path graph, and preferences for these students are depicted above them. 
Here, taller students are preferred, that is, $i_2 \succ i_4 \succ i_3 \succ i_1 \succ i_5$, resulting in a ``double-peaked'' preference.
It can be confirmed that this preference is a single-peaked preference on the tree decomposition that has width two shown in \cref{fig:tree-decomposition-of-two-peaked}.
Therefore, given such school preferences, it is possible to find a matching that satisfies PE, local ERF-1, and MB using our B-LT(3) mechanism.
In this way, 
we believe that our mechanism can handle various situations, even when the preferences of schools are not single-peaked on a path (or a tree), allowing schools to evaluate students based on multiple criteria.

\begin{figure}[t]
\centering
\includegraphics[scale = 0.3]{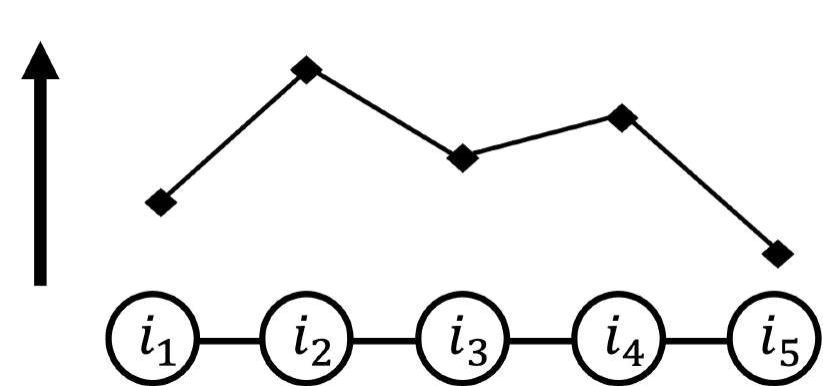}
\hspace{5cm}
\caption{A double-peaked preference}
\label{fig:two-peaked-preferences}
\hspace{10cm}
\end{figure}

\begin{figure}[t]
\centering
\includegraphics[scale = 0.3]{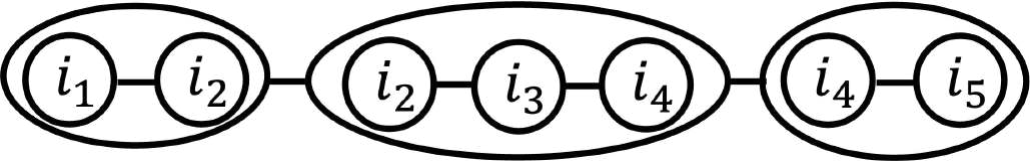}
\hspace{5cm}
\caption{A tree decomposition of a path that has width two}
\label{fig:tree-decomposition-of-two-peaked}
\hspace{10cm}
\end{figure}

Future work includes a clarification of what happens when we allow a one-way knowing relationship.
In this case, the graph representing the acquaintance relationships among students becomes a directed graph.
Also, when the graph has treewidth $k$ and schools' preferences are single-peaked on the graph, we only proposed a PE mechanism that bounds the number of the students by whom each student is locally envied, i.e., one that guarantees local ERF($k-1$).
It is open whether there exists a PE mechanism or even matching that bounds the number of students to whom each student has local envy, i.e., one that guarantees local EF($k-1$).



\begin{ack}
We would like to thank anonymous reviewers for their valuable comments. 
This work was partially supported by JST ERATO Grant Number JPMJER2301, and JSPS
KAKENHI Grant Numbers JP21H04979 and JP25K03186, Japan.
\end{ack}



\bibliography{matching}

\begin{thebibliography}{42}
\providecommand{\natexlab}[1]{#1}
\providecommand{\url}[1]{\texttt{#1}}
\expandafter\ifx\csname urlstyle\endcsname\relax
  \providecommand{\doi}[1]{doi: #1}\else
  \providecommand{\doi}{doi: \begingroup \urlstyle{rm}\Url}\fi

\bibitem[Abdulkadiro{\u{g}}lu and Grigoryan(2021)]{AG20}
A.~Abdulkadiro{\u{g}}lu and A.~Grigoryan.
\newblock Efficient and envy minimal matching.
\newblock \emph{Available at SSRN 3646475}, 2021.

\bibitem[Abdulkadiro{\u{g}}lu and S{\"o}nmez(1998)]{AS98}
A.~Abdulkadiro{\u{g}}lu and T.~S{\"o}nmez.
\newblock Random serial dictatorship and the core from random endowments in house allocation problems.
\newblock \emph{Econometrica}, 66\penalty0 (3):\penalty0 689--701, 1998.
\newblock \doi{https://doi.org/10.2307/2998580}.

\bibitem[Abdulkadiro{\u{g}}lu and S{\"o}nmez(2003)]{AS03}
A.~Abdulkadiro{\u{g}}lu and T.~S{\"o}nmez.
\newblock School choice: A mechanism design approach.
\newblock \emph{American economic review}, 93\penalty0 (3):\penalty0 729--747, 2003.

\bibitem[Abebe et~al.(2017)Abebe, Kleinberg, and Parkes]{Abebe2017}
R.~Abebe, J.~Kleinberg, and D.~C. Parkes.
\newblock Fair division via social comparison.
\newblock In \emph{Proceedings of the 16th Conference on Autonomous Agents and MultiAgent Systems}, page 281^^e2^^80^^93289. International Foundation for Autonomous Agents and Multiagent Systems, 2017.
\newblock \doi{10.5555/3091125.3091171}.

\bibitem[Abraham et~al.(2004)Abraham, Cechl{\'a}rov{\'a}, Manlove, and Mehlhorn]{ACMM04}
D.~J. Abraham, K.~Cechl{\'a}rov{\'a}, D.~F. Manlove, and K.~Mehlhorn.
\newblock Pareto optimality in house allocation problems.
\newblock In \emph{International symposium on algorithms and computation}, pages 3--15. Springer, 2004.
\newblock \doi{https://doi.org/10.1007/978-3-540-30551-4_3}.

\bibitem[Anshelevich et~al.(2013)Anshelevich, Bhardwaj, and Hoefer]{ABH13}
E.~Anshelevich, O.~Bhardwaj, and M.~Hoefer.
\newblock Friendship and stable matching.
\newblock In \emph{Algorithms--ESA 2013: 21st Annual European Symposium, Sophia Antipolis, France, September 2-4, 2013. Proceedings 21}, pages 49--60. Springer, 2013.
\newblock \doi{https://doi.org/10.1007/978-3-642-40450-4_5}.

\bibitem[Arcaute and Vassilvitskii(2009)]{Arcaute2009}
E.~Arcaute and S.~Vassilvitskii.
\newblock Social networks and stable matchings in the job market.
\newblock In S.~Leonardi, editor, \emph{Internet and Network Economics}, pages 220--231. Springer Berlin Heidelberg, 2009.
\newblock \doi{https://doi.org/10.1007/978-3-642-10841-9_21}.

\bibitem[Askalidis et~al.(2013)Askalidis, Immorlica, Kwanashie, Manlove, and Pountourakis]{AIK+13}
G.~Askalidis, N.~Immorlica, A.~Kwanashie, D.~F. Manlove, and E.~Pountourakis.
\newblock Socially stable matchings in the hospitals/residents problem.
\newblock In \emph{Workshop on Algorithms and Data Structures}, pages 85--96. Springer, 2013.
\newblock \doi{https://doi.org/10.1007/978-3-642-40104-6_8}.

\bibitem[Aziz et~al.(2019)Aziz, Gaspers, Sun, and Walsh]{Haris19matching}
H.~Aziz, S.~Gaspers, Z.~Sun, and T.~Walsh.
\newblock From matching with diversity constraints to matching with regional quotas.
\newblock In \emph{Proceedings of the 18th International Conference on Autonomous Agents and MultiAgent Systems (AAMAS-2019)}, pages 377--385, 2019.

\bibitem[Bei et~al.(2017)Bei, Qiao, and Zhang]{Bei2017}
X.~Bei, Y.~Qiao, and S.~Zhang.
\newblock Networked fairness in cake cutting.
\newblock \emph{CoRR}, abs/1707.02033, 2017.
\newblock \doi{10.48550/arXiv.1707.02033}.

\bibitem[Beynier et~al.(2019)Beynier, Chevaleyre, Gourv^^c3^^a8s, Harutyunyan, Lesca, Maudet, and Wilczynski]{Beynier2019}
A.~Beynier, Y.~Chevaleyre, L.~Gourv^^c3^^a8s, A.~Harutyunyan, J.~Lesca, N.~Maudet, and A.~Wilczynski.
\newblock Local envy-freeness in house allocation problems.
\newblock In \emph{Autonomous Agents and Multi-Agent Systems}, page 591^^e2^^80^^93627, 2019.
\newblock \doi{10.1007/s10458-019-09417-x}.

\bibitem[Bredereck et~al.(2022)Bredereck, Kaczmarczyk, and Niedermeier]{Bredereck}
R.~Bredereck, A.~Kaczmarczyk, and R.~Niedermeier.
\newblock Envy-free allocations respecting social networks.
\newblock \emph{Artificial Intelligence}, 305:\penalty0 103664, 2022.
\newblock \doi{https://doi.org/10.1016/j.artint.2022.103664}.

\bibitem[Cerrone et~al.(2024)Cerrone, Hermstr{\"u}wer, and Kesten]{CHK24}
C.~Cerrone, Y.~Hermstr{\"u}wer, and O.~Kesten.
\newblock School choice with consent; an experiment.
\newblock \emph{The Economic Journal}, 134:\penalty0 1760--1805, 2024.
\newblock \doi{https://doi.org/10.1093/ej/uead120}.

\bibitem[Cheng and McDermid(2013)]{CM13}
C.~T. Cheng and E.~McDermid.
\newblock Maximum locally stable matchings.
\newblock \emph{Algorithms}, 6\penalty0 (3):\penalty0 383--395, 2013.
\newblock \doi{https://doi.org/10.3390/a6030383}.

\bibitem[Cho et~al.(2024)Cho, Kimura, Liu, Liu, Liu, Sun, Yahiro, and Yokoo]{cho:2024}
S.-H. Cho, K.~Kimura, K.~Liu, K.~Liu, Z.~Liu, Z.~Sun, K.~Yahiro, and M.~Yokoo.
\newblock Fairness and efficiency trade-off in two-sided matching.
\newblock In \emph{Proceedings of the 23rd International Conference on Autonomous Agents and MultiAgent Systems (AAMAS-2024)}, pages 372--380, 2024.

\bibitem[Cygan et~al.(2015)Cygan, Fomin, Kowalik, Lokshtanov, Marx, Pilipczuk, Pilipczuk, and Saurabh]{CFK+15}
M.~Cygan, F.~V. Fomin, {\L}.~Kowalik, D.~Lokshtanov, D.~Marx, M.~Pilipczuk, M.~Pilipczuk, and S.~Saurabh.
\newblock \emph{Parameterized algorithms}.
\newblock Springer, 2015.
\newblock \doi{https://doi.org/10.1007/978-3-319-21275-3}.

\bibitem[Demange(1982)]{DEMANGE1982389}
G.~Demange.
\newblock Single-peaked orders on a tree.
\newblock \emph{Mathematical Social Sciences}, 3\penalty0 (4):\penalty0 389--396, 1982.
\newblock \doi{https://doi.org/10.1016/0165-4896(82)90020-8}.

\bibitem[Dur et~al.(2019)Dur, Gitmez, and Y^^c4^^b1lmaz]{DGY19}
U.~Dur, A.~A. Gitmez, and {\"O}.~Y^^c4^^b1lmaz.
\newblock School choice under partial fairness.
\newblock \emph{Theoretical Economics}, 14:\penalty0 1309--1346, 2019.
\newblock \doi{https://doi.org/10.3982/TE2482}.

\bibitem[Elkind et~al.(2017)Elkind, Lackner, and Peters]{ELP17}
E.~Elkind, M.~Lackner, and D.~Peters.
\newblock Structured preferences.
\newblock \emph{Trends in computational social choice}, pages 187--207, 2017.

\bibitem[Endriss(2017)]{Endriss:2017}
U.~Endriss.
\newblock \emph{Trends in Computational Social Choice}.
\newblock Lulu.com, 2017.
\newblock \doi{10.5555/3180776}.

\bibitem[Flammini et~al.(2019)Flammini, Mauro, and Tonelli]{FLAMMINI20191}
M.~Flammini, M.~Mauro, and M.~Tonelli.
\newblock On social envy-freeness in multi-unit markets.
\newblock \emph{Artificial Intelligence}, 269:\penalty0 1--26, 2019.
\newblock \doi{https://doi.org/10.1016/j.artint.2018.12.003}.

\bibitem[Gale and Shapley(1962)]{Gale:AMM:1962}
D.~Gale and L.~S. Shapley.
\newblock College admissions and the stability of marriage.
\newblock \emph{The American Mathematical Monthly}, 69\penalty0 (1):\penalty0 9--15, 1962.
\newblock \doi{10.2307/2312726}.

\bibitem[Goto et~al.(2017)Goto, Kojima, Kurata, Tamura, and Yokoo.]{goto:17}
M.~Goto, F.~Kojima, R.~Kurata, A.~Tamura, and M.~Yokoo.
\newblock Designing matching mechanisms under general distributional constraints.
\newblock \emph{American Economic Journal: Microeconomics}, 9\penalty0 (2):\penalty0 226--262, 2017.
\newblock \doi{10.1257/mic.20160124}.

\bibitem[Hansen and Thisse(1981)]{HANSEN19811}
P.~Hansen and J.-F. Thisse.
\newblock Outcomes of voting and planning: Condorcet, weber and rawls locations.
\newblock \emph{Journal of Public Economics}, 16\penalty0 (1):\penalty0 1--15, 1981.
\newblock \doi{https://doi.org/10.1016/0047-2727(81)90039-6}.

\bibitem[Hoefer and Wagner(2017)]{HW17}
M.~Hoefer and L.~Wagner.
\newblock Locally stable marriage with strict preferences.
\newblock \emph{SIAM Journal on Discrete Mathematics}, 31\penalty0 (1):\penalty0 283--316, 2017.
\newblock \doi{https://doi.org/10.1137/16M1077283}.

\bibitem[Hosseini et~al.(2015)Hosseini, Larson, and Cohen]{hosseini2015manipulablity}
H.~Hosseini, K.~Larson, and R.~Cohen.
\newblock On manipulablity of random serial dictatorship in sequential matching with dynamic preferences.
\newblock In \emph{Proceedings of the 29th {AAAI} Conference on Artificial Intelligence (AAAI-2015)}, pages 4168--4169, 2015.
\newblock \doi{10.1609/aaai.v29i1.9744}.

\bibitem[Imamura and Kawase(2024)]{IK24}
K.~Imamura and Y.~Kawase.
\newblock Efficient matching under general constraints.
\newblock \emph{Games and Economic Behavior}, 145:\penalty0 197--207, 2024.
\newblock \doi{https://doi.org/10.1016/j.geb.2024.03.013}.

\bibitem[Ismaili et~al.(2019)Ismaili, Hamada, Zhang, Suzuki, and Yokoo]{IsmailiHZSY19}
A.~Ismaili, N.~Hamada, Y.~Zhang, T.~Suzuki, and M.~Yokoo.
\newblock Weighted matching markets with budget constraints.
\newblock \emph{Journal of Artificial Intelligence Research}, 65:\penalty0 393--421, 2019.
\newblock \doi{10.1613/jair.1.11582}.

\bibitem[Kawase and Iwasaki(2017)]{kawase2017near}
Y.~Kawase and A.~Iwasaki.
\newblock Near-feasible stable matchings with budget constraints.
\newblock In \emph{Proceedings of the 26th International Joint Conference on Artificial Intelligence (IJCAI-2017)}, pages 242--248, 2017.
\newblock \doi{10.24963/ijcai.2017/35}.

\bibitem[Kesten(2010)]{Kesten10}
O.~Kesten.
\newblock School choice with consent.
\newblock \emph{The Quarterly Journal of Economics}, 125:\penalty0 1297--1348, 2010.
\newblock \doi{https://doi.org/10.1162/qjec.2010.125.3.1297}.

\bibitem[Kobren et~al.(2019)Kobren, Saha, and McCallum]{KSM19}
A.~Kobren, B.~Saha, and A.~McCallum.
\newblock Paper matching with local fairness constraints.
\newblock In \emph{Proceedings of the 25th ACM SIGKDD Conference on Knowledge Discovery and Data Mining (KDD-2019)}, pages 1247--1257, 2019.

\bibitem[Maniu et~al.(2019)Maniu, Senellart, and Jog]{MSJ19}
S.~Maniu, P.~Senellart, and S.~Jog.
\newblock An experimental study of the treewidth of real-world graph data.
\newblock In \emph{ICDT 2019--22nd International Conference on Database Theory}, page~18, 2019.

\bibitem[Moulin(1980)]{Moulin80}
H.~Moulin.
\newblock On strategy-proofness and single peakedness.
\newblock \emph{Public Choice}, 35\penalty0 (4):\penalty0 437--455, 1980.

\bibitem[Roth(1982)]{roth1982economics}
A.~E. Roth.
\newblock The economics of matching: Stability and incentives.
\newblock \emph{Mathematics of Operations Research}, 7\penalty0 (4):\penalty0 617--628, 1982.
\newblock \doi{10.1287/moor.7.4.617}.

\bibitem[Roth(1986)]{Roth:1986}
A.~E. Roth.
\newblock On the allocation of residents to rural hospitals: a general property of two-sided matching markets.
\newblock \emph{Econometrica: Journal of the Econometric Society}, pages 425--427, 1986.

\bibitem[Roth and Sotomayor(1990)]{Roth:CUP:1990}
A.~E. Roth and M.~A.~O. Sotomayor.
\newblock \emph{Two-Sided Matching: A Study in Game-Theoretic Modeling and Analysis (Econometric Society Monographs)}.
\newblock Cambridge University Press, 1990.
\newblock \doi{10.1017/CCOL052139015X}.

\bibitem[Saban and Sethuraman(2015)]{SS15}
D.~Saban and J.~Sethuraman.
\newblock The complexity of computing the random priority allocation matrix.
\newblock \emph{Mathematics of Operations Research}, 40\penalty0 (4):\penalty0 1005--1014, 2015.
\newblock \doi{https://doi.org/10.1287/moor.2014.0707}.

\bibitem[Satterthwaite and Sonnenschein(1981)]{satterthwaite1981}
M.~A. Satterthwaite and H.~Sonnenschein.
\newblock Strategy-proof allocation mechanisms at differentiable points.
\newblock \emph{The Review of Economic Studies}, 48\penalty0 (4):\penalty0 587--597, 1981.

\bibitem[Suzuki et~al.(2023)Suzuki, Tamura, Yahiro, Yokoo, and Zhang]{suzuki2022strategyproof}
T.~Suzuki, A.~Tamura, K.~Yahiro, M.~Yokoo, and Y.~Zhang.
\newblock Strategyproof allocation mechanisms with endowments and {M}-convex distributional constraints.
\newblock \emph{Artificial Intelligence}, 315:\penalty0 103825, 2023.
\newblock \doi{10.1016/j.artint.2022.103825}.

\bibitem[Toda(2006)]{Toda2006}
M.~Toda.
\newblock Monotonicity and consistency in matching markets.
\newblock \emph{International Journal of Game Theory}, 34:\penalty0 13--31, 2006.
\newblock \doi{10.1007/s00182-005-0002-5}.

\bibitem[Trick(1989)]{TRICK1989329}
M.~A. Trick.
\newblock Recognizing single-peaked preferences on a tree.
\newblock \emph{Mathematical Social Sciences}, 17\penalty0 (3):\penalty0 329--334, 1989.
\newblock \doi{https://doi.org/10.1016/0165-4896(89)90060-7}.

\bibitem[Yahiro et~al.(2020)Yahiro, Zhang, Barrot, and Yokoo]{Yahiro18}
K.~Yahiro, Y.~Zhang, N.~Barrot, and M.~Yokoo.
\newblock Strategyproof and fair matching mechanism for ratio constraints.
\newblock \emph{Autonomous Agents and Multi-Agent Systems}, 34:\penalty0 1--29, 2020.
\newblock \doi{10.1007/s10458-020-09448-9}.

\end{thebibliography}

\newpage
\appendix

\section{Comparison between local stability and local envy-freeness}
\label{sec:Comparison-LS-and-LEF}
In this section, we compare the locally stable (LS) matchings, as proposed by \citet{Arcaute2009}, 
with the locally envy-free (LEF) matchings introduced in this paper. 
Specifically, we clarify the inclusion relationship between the set of LS matchings and the set of LEF matchings, as illustrated in \cref{fig:Comparison of LS and LEF}.

\begin{figure}[t]
\centering
\includegraphics[scale = 0.3]{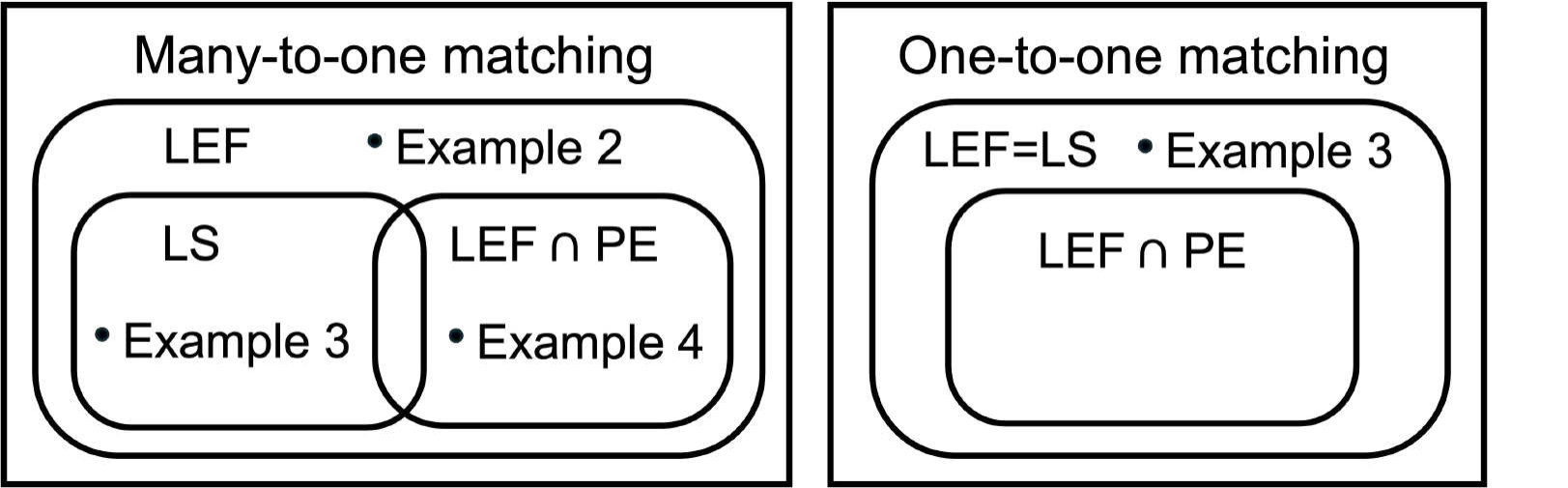}
\hspace{5cm}
\caption{Inclusion relationship between the set of LS matchings and the set of LEF matchings. Here, ``LEF'' represents the set of LEF matchings, ``LS'' represents the set of LS matchings, and ``PE'' represents the set of PE matchings.}
\label{fig:Comparison of LS and LEF}
\hspace{10cm}
\end{figure}

First, we define local stable (LS) matching, as defined by \citet{Arcaute2009} in the context of the job market, 
in our setting.
\begin{definition}[locally stable matching]
\label{def:locally-stable}
Let $G = (\studentset,E)$ be a student acquaintance graph.
We say that a matching $\matching$ is a locally stable (LS) matching with respect to $G$ if, for all $\student \in \studentset$ and $\school \in \schoolset$, $(\student, \school)$ is a blocking pair if and only if $N(\student) \cap \matching_\school = \emptyset$ (i.e., no neighboring students are matched to school $\school$).
In this context, we say that $(\student,\school)$ is a blocking pair if and only if $\school \succ_\student \matching_\student$ and $\student \succ_\school \min(\matching_\school)$, where $\min(\matching_\school)$ is the least preferred student matched to school $\school$ (with respect to school $\school$'s preference) if $|\matching_\school| = q_{\school}$, and $\min(\matching_\school) = \emptyset$ otherwise.
\end{definition}

Next, we show that LS matching and LEF matching coincide in one-to-one matching, that is, when the schools' maximum quotas are all one.

\begin{theorem}
\label{thm:equivalence-ls-lef}
When $|q_\school| = 1$ for all school $\school$, the set of locally stable matchings is identical to the set of locally envy-free matchings.
\end{theorem}
\begin{proof}
First, 
we show that any LS matching must satisfy LEF, by contraposition.
That is, 
we prove that if a matching does not satisfy LEF, 
then it cannot be LS.
We consider the case where $(\student', \school) \in Y, \school \succ_{\student} \matching_{\student}$, and $\student \succ_\school \matching_{\school}$.
In other words, student $\student$ has envy toward a neighboring student $\student'$ who is assigned to a school $\school$, then the matching violations LEF. 
In such a case, since $\school \succ_{\student} \matching_{\student}$ and $\student \succ_\school \matching_{\school}$, the pair $(\student, \school)$ constitutes a blocking pair, and $N(\student) \cap \matching_\school$ is nonempty as it contains $\student'$.
Hence, this matching is not LS.
By contraposition, 
we conclude that if a matching is LS, 
then it must also satisfy LEF.

Next, we show that if a matching satisfies LEF, then it also satisfies LS in the case where the capacity of each school is one, by contraposition.
A matching is not LS if there exists a pair $(\student, \school)$ that forms a blocking pair and $N(\student) \cap \matching_\school \neq \emptyset$.
This means that $\school \succ_\student \matching_\student$, $\student \succ_\school \matching_\school$.
Since the quota of $\school$ is one, there exists $\student'$ such that 
$\matching_\school = \{i'\}$ and $\student' \in N(\student)$.
Hence, $\student$ has a justified envy towards a neighboring student $\student'$ and $\matching$ is not LEF.

Therefore, under the assumption that each school has a capacity of one, the set of LS matchings coincides with the set of LEF matchings.
\end{proof}

Next, we show that an LS matching is always LEF, however, 
LS is in general a strictly stronger property than LEF when some school has quota more than one.

\begin{theorem}
\label{thm:ls_stronger_than_lef}
When there exists a school $\school$ with quota $|q_\school| \geq 2$, local stability is a strictly stronger property than local envy-freeness.
\end{theorem}
\begin{proof}
First, we show that every LS matching satisfies LEF. 
This fact has already been established in the proof of \cref{thm:equivalence-ls-lef}.

Next, we provide an example (Example~\ref{exa:lef-to-ls}) below with a specific preference profile and student acquaintance graph, where a matching satisfies LEF but does not satisfy LS, thereby demonstrating that not all LEF matchings are necessarily LS matchings.
\end{proof}

\begin{example}
\label{exa:lef-to-ls}
Let $\studentset = \{\student_1, \student_2\}$, $\schoolset = \{\school_1, \school_2\}$, and $q_{\school_1} = 1, q_{\school_2} = 2$.
Assume that the two students are connected by an edge in the student acquaintance graph.
Assume that $\succ_\student$ and $\succ_\school$ are given as follows.
\begin{align*}
\student_1: \school_2 \succ_{\student_1} \school_1 \qquad \school_1: \student_1 \succ_{\school_1} \student_2\\
\student_2: \school_2 \succ_{\student_2} \school_1 \qquad \school_2: \student_2 \succ_{\school_2} \student_1
\end{align*}
Consider the matching $\matching = \{(\student_1, \school_1), (\student_2, \school_2)\}$.
One can confirm that $\matching$ is LEF.
However, 
since $(\student_1, \school_2)$ is a blocking pair and $N(\student_1) \cap \matching_{\school_2} = \{ \student_2 \} \neq \emptyset$, 
$\matching$ is not LS.
Moreover, $\matching$ is Pareto dominated by $\matching' = \{(\student_1, \school_2), (\student_2, \school_2)\}$, and thus $\matching$ is not Pareto efficient.
Therefore, $\matching$ is LEF but is not LS nor PE.
\end{example}

\begin{theorem}
\label{thm:lefpe-independent-ls}
The set of matchings that are both LEF and PE is independent from the set of LS matchings.
\end{theorem}
\begin{proof}
To prove this theorem, we will demonstrate two things:
\begin{enumerate}
    \item An LS matching does not necessarily satisfy either LEF or PE.
    \item A matching that satisfies both LEF and PE does not necessarily satisfy LS.
\end{enumerate}
First, to demonstrate that not LS matchings necessarily satisfy LEF and PE, we present an example (Example~\ref{exa:ls-to-lefpe}) below.

Next, 
we present an example (Example~\ref{exa:lefpe-to-ls}) below based on a specific preference profile and student acquaintance graph.
This example demonstrates a matching that satisfies both LEF and PE, but not LS.
This shows that not all matchings that satisfy both LEF and PE necessarily satisfy LS.
\end{proof}

\begin{example}
\label{exa:ls-to-lefpe}
We consider the empty matching $\matching = \emptyset$.
In this case, the condition $N(\student) \cap \matching_\school = \emptyset$ trivially holds for all $\student$ and $\school$, and therefore $\matching$ satisfies LS.
However, if there exists at least one student who prefers being matched to some school with a remaining seat, then $\matching$ does not satisfy PE.
\end{example}

\begin{example}
\label{exa:lefpe-to-ls}
Let $\studentset = \{\student_1, \student_2, \student_3\}$, $\schoolset = \{\school_1, \school_2\}$, and $q_{\school_1} = 1, q_{\school_2} = 2$.
Assume that the student acquaintance graph is a path graph depicted in~\cref{fig:path} in~\cref{sec:failure-fundamental-theorems}.
Assume that $\succ_\student$ and $\succ_\school$ are given as follows.
\begin{align*}
&\student_1: \school_2 \succ_{\student_1} \school_1 \qquad \school_1: \student_1 \succ_{\school_1} \student_2 \succ_{\school_1} \student_3\\
&\student_2: \school_2 \succ_{\student_2} \school_1 \qquad \school_2: \student_2 \succ_{\school_2} \student_1 \succ_{\school_2} \student_3\\
&\student_3: \school_2 \succ_{\student_3} \school_1 
\end{align*}
We consider the matching $\matching = \{(\student_1, \school_1), (\student_2, \school_2), (\student_3, \school_2)\}$.

Although $\student_1$ prefers $\school_2$ to her assigned school $\school_1$,
$\student_1$ does not experience local envy because $\school_2$ prefers $\student_2 \in N(\student_1)$ to $\student_1$.
Students $\student_2$ and $\student_3$ are matched with their most preferred schools, and therefore do not experience any local envy.
Thus, $\matching$ satisfies LEF.
Moreover, since all students prefer $\school_2$ to $\school_1$, the matching $\matching$ is also PE.
On the other hand, since 
$\school_2 \succ_{\student_1} \school_1$ and $\student_1 \succ_{\school_2} \student_3$, 
the pair $(\student_1, \school_2)$ is a blocking pair.
Moreover, 
$N(\student_1) \cap \matching_{\school_2} = \{ \student_2 \} \neq \emptyset$.
Therefore, $\matching$ is not an LS matching.
\end{example}

\begin{remark}
In this section, we have used PE as an indicator of efficiency, but it can be seen from the examples in Figure \ref{fig:Comparison of LS and LEF} that the inclusion relationship holds even if PE is replaced with nonwastefulness (NW). 
\end{remark}

\end{document}